\documentclass[english]{article}

\usepackage{fontenc}
\usepackage{geometry}
\usepackage{array}

\usepackage{fullpage}
\usepackage{multirow} 
\usepackage{cite}
\usepackage{amsfonts}
\usepackage{amssymb}
\usepackage{amstext}
\usepackage{amsmath}
\usepackage{mathtools}
\usepackage{paralist}

\usepackage{tikz}
\usetikzlibrary{calc}
\usetikzlibrary{shapes,arrows,positioning,shadows,snakes}

\usepackage{footmisc}

\usepackage{amsthm} % NEED TO REMOVE amsthm in ACM EC Format

\usepackage{pifont}
\usepackage[hang,small,bf]{caption}
\usepackage{subcaption}

 \providecommand{\tabularnewline}{\\}

\usepackage{nameref}
\usepackage[linktocpage=true,pagebackref=true]{hyperref}
\usepackage{cleveref}

\usepackage{thmtools,thm-restate} % See section 1.4 of the pdf above

\usepackage{graphicx}
\usepackage{graphics}
\usepackage{colordvi}
\usepackage{xspace}
\usepackage{algorithm}
\usepackage{algorithmicx}
\usepackage{algpseudocode}
\usepackage{url}
\usepackage{enumitem}
\usepackage{authblk}

\newcommand{\greedy}{{\sc Greedy}\xspace} 
 
\newcommand{\sgreedy}{{\sc SGreedy}\xspace} 
\newcommand{\greedyleft}{{\sc GreedyLeft}\xspace} 
\newcommand{\greedyright}{{\sc GreedyRight}\xspace} 
\newcommand{\stair}{{\sf stair}\xspace} 
\newcommand{\stairleft}{{\textsf{stairLeft}}\xspace} 
\newcommand{\stairright}{{\textsf{stairRight}}\xspace} 

\newcommand{\agreedy}{{\sc RGreedy}\xspace}

\newenvironment{proofof}[1]{\noindent{\bf Proof of #1.}}%
        {\hspace*{\fill}$\Box$\par\vspace{4mm}}

\makeatletter
\def\thmt@refnamewithcomma #1#2#3,#4,#5\@nil{%
  \@xa\def\csname\thmt@envname #1utorefname\endcsname{#3}%
  \ifcsname #2refname\endcsname
    \csname #2refname\expandafter\endcsname\expandafter{\thmt@envname}{#3}{#4}%
  \fi
}
\makeatother

\declaretheorem[numberwithin=section,refname={Theorem,Theorems},Refname={Theorem,Theorems}]{theorem}
\declaretheorem[numberlike=theorem,refname={Lemma,Lemmas},Refname={Lemma,Lemmas}]{lemma}

\declaretheorem[numberlike=theorem,refname={Corollary,Corollaries},Refname={Corollary,Corollaries}]{corollary}
\declaretheorem[numberlike=theorem,refname={Proposition,Propositions},Refname={Proposition,Propositions}]{proposition}

\declaretheorem[numberlike=theorem,refname={observation,observations},Refname={Observation,Observations}]{observation}

\declaretheorem[numberlike=theorem,refname={Claim, Claims},Refname={Claim, Claims}]{claim}

\declaretheorem[numberlike=theorem]{definition}

\newtheorem{conjecture}{Conjecture}
\newcommand{\abs}[1]{\lvert #1\rvert}

\newcommand{\OPT}{\mbox{\sf OPT}}

\newcommand{\opt}{\mbox{\sf OPT}}

\newcommand{\set}[1]{\left\{ #1 \right\}}

\newcommand{\tset}{{\mathcal T}}

\newcommand{\pset}{{\mathcal{P}}}

\newcommand{\bset}{{\mathcal{B}}}
\newcommand{\aset}{{\mathcal{A}}} \newcommand{\cA}{{\mathcal{A}}} \newcommand{\cB}{{\mathcal{B}}}
\newcommand{\cset}{{\mathcal{C}}}

\newcommand{\sset}{{\mathcal{S}}}

\newcommand{\be}{\begin{enumerate}}
\newcommand{\ee}{\end{enumerate}}
\newcommand{\bd}{\begin{description}}
\newcommand{\ed}{\end{description}}
\newcommand{\bi}{\begin{itemize}}
\newcommand{\ei}{\end{itemize}}

\renewcommand{\phi}{\varphi}

\newcommand{\N}{\ensuremath{\mathbb N}}

\newcommand{\wing}{\mathit{wing}}

\newcommand{\inc}{\texttt{Inc}}
\newcommand{\dec}{\texttt{Dec}}
\newcommand{\capture}{\texttt{Cap}}
\newcommand{\alt}{\texttt{Alt}}
\newcommand{\Greedy}{\textsc{Greedy}} 
 \newcommand{\GreedyR}{\textsc{GreedyRight}}
\newcommand{\merge}{\mathit{merge}}
\newcommand{\Split}{\mathit{split}}
\newcommand{\tP}{\tilde{P}}
\newcommand{\tY}{Y_{\tP}} 
\newcommand{\topwing}{\mathit{topwing}}
\newcommand{\ymin}{\mathit{ymin}}
\newcommand{\ymax}{\mathit{ymax}}
\newcommand{\xmin}{\mathit{xmin}}
\newcommand{\xmax}{\mathit{xmax}}
\newcommand{\junk}{\mathit{junk}}
\newcommand{\mleft}{\mathit{left}}
\newcommand{\mright}{\mathit{right}}
\newcommand{\degree}{\mathit{degree}}
\newcommand{\parent}{\mathit{parent}}

\def\markatright#1{\leavevmode\unskip\nobreak\quad\hspace*{\fill}{#1}}
\renewenvironment{proof}
  {\begin{trivlist}\item[\hskip\labelsep{\emph{Proof}.}]}
  {\markatright{\qed}\end{trivlist}}

\def\ShowComment{True}

\ifdefined\ShowComment

\def\parinya#1{\marginpar{$\leftarrow$\fbox{P}}\footnote{$\Rightarrow$~{\sf #1 --Parinya}}}

\else

\def\parinya#1{}

\fi

\ifdefined\ShowComment

\def\kurt#1{\marginpar{$\leftarrow$\fbox{K}}\footnote{$\Rightarrow$~{\sf #1 --Kurt}}}

\else

\def\kurt#1{}

\fi

\ifdefined\ShowComment

\def\thatchaphol#1{\marginpar{$\leftarrow$\fbox{T}}\footnote{$\Rightarrow$~{\sf #1 --Thatchaphol}}}

\else

\def\thatchaphol#1{}

\fi

\ifdefined\ShowComment

\def\laszlo#1{\marginpar{$\leftarrow$\fbox{L}}\footnote{$\Rightarrow$~{\sf #1 --LK}}}

\else

\def\laszlo#1{}

\fi

\ifdefined\ShowComment

\else

\def\laszlo#1{}

\fi

\newcommand{\comment}[1]{\textcolor{blue}{\textbf{#1}}}
\newcommand{\lk}[1]{\comment{\textcolor{blue}{[LK:] #1}}}

\title{Pattern-avoiding access in binary search trees} 
\date{}

\author[1]{Parinya Chalermsook} 
\author[1]{Mayank Goswami}  
\author[2]{L\'{a}szl\'{o} Kozma} 
\author[1]{Kurt Mehlhorn} 
\author[3]{Thatchaphol Saranurak\thanks{Work partly done while at Saarland University.}}  

\affil[1]{{\footnotesize Max-Planck Institute for Informatics, Germany.}  
{\footnotesize \tt \{parinya,gmayank,mehlhorn\}@mpi-inf.mpg.de} }
\affil[2]{{\footnotesize Saarland University, Germany. } {\footnotesize \tt kozma@cs.uni-saarland.de}} 
\affil[3]{{\footnotesize KTH Royal Institute of Technology, Sweden.} {\footnotesize \tt thasar@kth.se} } 

\begin{document}

\begin{titlepage}
\maketitle
\pagenumbering{roman}

\begin{abstract}
The {\em dynamic optimality conjecture} is perhaps the most fundamental open question about binary search trees (BST). It postulates the existence of an asymptotically optimal online BST, i.e.\ one that is {\em constant factor competitive} with any BST on any input access sequence. The two main candidates for dynamic optimality in the literature are {\em splay trees} [Sleator and Tarjan, 1985], and \greedy [Lucas, 1988; Munro, 2000; Demaine et al.\ 2009]. Despite BSTs being among the simplest data structures in computer science, and despite extensive effort over the past three decades, the conjecture remains elusive. Dynamic optimality is trivial for \emph{almost all} sequences: the optimum access cost of most 
length-$n$ sequences is $\Theta(n \log n)$, 
achievable by any balanced BST. 
Thus, the obvious missing step towards the conjecture is an understanding of the ``easy'' access sequences, and indeed the most fruitful research direction so far has been the study of specific sequences, whose ``easiness'' is captured by a parameter of interest.
For instance, splay provably achieves the bound of $O(n d)$ when $d$ roughly measures the distances between consecutive accesses (dynamic finger), the average entropy (static optimality), or the delays between multiple accesses of an element (working set). 
The difficulty of proving dynamic optimality is witnessed by other highly restricted special cases that remain unresolved; one prominent example is the {\em traversal conjecture} [Sleator and Tarjan, 1985], which states that {\em preorder sequences} (whose optimum is linear) are linear-time accessed by splay trees; no online BST is known to satisfy this conjecture.  

In this paper, we prove two different relaxations of the traversal conjecture for \greedy: (i) {\sc Greedy} is almost linear for preorder traversal, %(the bound involves the inverse Ackermann function) 
(ii) if a linear-time preprocessing\footnote{The purpose of preprocessing is to bring the data structure into a state of our choice. This state is independent
of the actual input. This kind of preprocessing is implicit in the Demaine et al.\ definition of \greedy.} is allowed, {\sc Greedy} is in fact linear.
These statements are corollaries of our more general results that express the complexity of access sequences in terms of a {\em pattern avoidance} parameter $k$. 
Pattern avoidance is a well-established concept in combinatorics, and the classes of input sequences thus defined are rich, e.g.\ the $k=3$ case includes preorder sequences. For any sequence $X$ with parameter $k$, our most general result shows that \greedy achieves the cost $n 2^{\alpha(n)^{O(k)}}$ where $\alpha$ is the inverse Ackermann function. 
Furthermore, a broad subclass of parameter-$k$ sequences has a natural combinatorial interpretation as $k$-{\em decomposable sequences}. For this class of inputs, we obtain an $n 2^{O(k^2)}$ bound for \greedy when preprocessing is allowed. 
For $k=3$, these results imply (i) and (ii). 
To our knowledge, these are the first upper bounds for \greedy that are not known to hold for any other online BST.
To obtain these results we identify an {\em input-revealing} property of \greedy. Informally, this means that the execution log partially reveals the structure of the access sequence. This property facilitates the use of rich technical tools from {\em forbidden submatrix theory}.

Further studying the intrinsic complexity of $k$-decomposable sequences, we make several observations. 
First, in order to obtain an offline optimal BST, it is enough to bound \greedy on non-decomposable access sequences. 
Furthermore, we show that the optimal cost for $k$-decomposable sequences is $\Theta(n \log k)$, which is well below the proven performance of all known BST algorithms. Hence, sequences in this class can be seen as a ``candidate counterexample'' to dynamic optimality.

\end{abstract} 

\newpage
\end{titlepage}

\newpage
\pagenumbering{arabic}

\section{Introduction}
The binary search tree (BST) model is one of the most fundamental and thoroughly studied data access models in computer science. 
In this model, given a sequence $X \in [n]^m$ of keys, we are interested in $\OPT(X)$, the optimum cost of accessing $X$ by a binary search tree. When the tree does not change between accesses, $\OPT$ is well understood: both efficient exact algorithms (Knuth~\cite{knuth_optimum}) and a linear time approximation (Mehlhorn~\cite{mehlhorn1975nearly}) have long been known. By contrast, in the dynamic BST model (where restructuring of the tree is allowed between accesses) our understanding of $\OPT$ is very limited. No polynomial time exact or constant-approximation algorithms are known for computing $\OPT$ in the dynamic BST model.

Theoretical interest in the dynamic BST model is partly due to the fact that it holds the possibility of a meaningful ``instance optimality'', i.e.\ of a BST algorithm that is constant factor competitive (in the amortized sense) with any other BST algorithm, even with those that are allowed to see into the future. 
The existence of such a ``dynamically optimal'' algorithm has been postulated in 1985 by Sleator and Tarjan~\cite{ST85}; they conjectured their splay tree to have this property. 
Later, Lucas~\cite{Luc88} and independently Munro~\cite{Mun00} proposed a greedy offline algorithm as another candidate for dynamic optimality. 
In 2009 Demaine, Harmon, Iacono, Kane, and P\u{a}tra\c{s}cu (DHIKP)~\cite{DHIKP09} presented a geometric view of the BST model, in which the Lucas-Munro offline algorithm naturally turns into an online one (simply called {\sc Greedy}).  
Both splay and \greedy are considered plausible candidates for achieving optimality. However, despite decades of research, neither of the two algorithms are known to be $o(\log n)$-competitive (the best known approximation ratio for any BST is $O(\log \log n)$~\cite{tango}, but the algorithms achieving this bound are less natural). Thus, the dynamic optimality conjecture remains wide open.

Why is dynamic optimality difficult to prove (or disprove)? For a vast majority of access sequences, the optimal access cost by any dynamic BST (online or offline) is\footnote{This observation appears to be folklore. We include a proof in Appendix~\ref{sec:random-hard}.} $\Theta(m \log n)$ 
and any static balanced tree can achieve this bound.\footnote{As customary, we only consider successful accesses, we ignore inserts and deletes, and we assume that $m \geq n$. }
Hence, dynamic optimality is almost trivial on most sequences. However, when an input sequence has complexity $o(m \log n)$, a candidate BST must ``learn'' and ``exploit'' the specific structure of the sequence, in order to match the optimum. This observation motivates the popular research direction that aims at understanding ``easy'' sequences, i.e.\ those with $\opt = O(m d)$ when $d$ is a parameter that depends on the structure of the sequence; typically $d=o(\log n)$. 
Parameter $d$ can be seen as a measure of ``easiness'', and the objective is to prove that a candidate algorithm achieves the absolute bound of $O(m d)$; note that this does not immediately imply that the algorithm matches the optimum.

The splay tree is known to perform very well on many ``easy'' sequences: It achieves the access time of $O(m d)$ when $d$ is a parameter that roughly captures the average distances between consecutive accesses~\cite{finger1,finger2}, the average entropy, or the recentness of the previous accesses~\cite{ST85}; these properties are called {\em dynamic finger}, {\em static optimality} and {\em working set} bounds respectively (the latter two bounds are also satisfied by \greedy~\cite{Fox11}). %Furthermore, both splay and \greedy can access the keys in a monotone sequence in the optimal $O(n)$ time~\cite{tarjan_sequential, Fox11} ({\em sequential access}). 
The notorious difficulty of dynamic optimality is perhaps witnessed by the fact that many other highly restricted special cases have resisted attempts and remain unresolved. These special cases stand as a roadblock to proving (or disproving) the optimality of any algorithm: proving that $\mathcal{A}$ is optimal requires proving that it performs well on the easy inputs; refuting the optimality of $\mathcal{A}$ requires a special ``easy'' subclass %i.e.\ one with $\opt = o(m \log n)$, 
on which $\mathcal{A}$ is provably inefficient. In any case, a better understanding of easy sequences is needed, but currently lacking, as observed by several authors~\cite{DHIKP09,deque_Pet08}.%% (``there is currently no {\em theory} of access sequences whose inherent complexity is linear.'')

One famous example of such roadblocks is the {\em traversal conjecture}, proposed by Sleator and Tarjan in 1985.  
This conjecture states that starting with an arbitrary tree $T$ (called {\em initial tree}) and accessing its keys using the splay algorithm in the order given by the preorder sequence of another tree $T'$ takes linear time.
Since the optimum of such preorder sequences is linear, any dynamically optimal algorithm must match this bound.  
After three decades, only very special cases of this conjecture are known to hold, namely when $T'$ is a monotone path (the so-called \emph{sequential access}~\cite{tarjan_sequential})
or when $T'= T$~\cite{chaudhuri}.\footnote{This result implies an offline BST with $O(n)$ cost for accessing a preorder sequence, by first rotating $T$ to $T'$.}
Prior to this work, no online BST algorithm was known to be linear on such access sequences.

\begin{table} 
\scriptsize\centering 
\begin{tabular}{|>{\raggedright}p{0.08\paperwidth}|>{\raggedright}p{0.15\paperwidth}|>{\centering}p{0.2\paperwidth}|>{\centering}p{0.16\paperwidth}|>{\centering}p{0.13\paperwidth}|}
\hline 
 & Structure  & splay bound & \greedy bound & Remark \tabularnewline
\hline 
\hline 
Static optimality & low entropy & $O(\sum_{i=1}^{n}f_{i}(1+\log\frac{m}{f_{i}}))$ \\~\cite{ST85}& same as splay ~\cite{Fox11, our_icalp} & \multirow{2}{0.1\paperwidth}{corollaries of Access~Lemma \cite{ST85}} \tabularnewline
\cline{1-4} 
Working set  & temporal locality & $O(m+n\log n+\sum_{t=1}^{m}\log(\tau_{t}+1))$\\ ~\cite{ST85}&  same as splay ~\cite{Fox11, our_icalp}& \tabularnewline
\hline 
Dynamic finger  & key-space locality & $O(m+\sum_{t=2}^{m}\log|x_{t}-x_{t-1}|)$\\~\cite{finger1,finger2} & ? & \tabularnewline
\hline
Deque  & $n$ updates at min/max elements & $O(n\alpha^{*}(n))$ ~\cite{deque_Pet08} & $O(n2^{\alpha(n)})$ ~\cite{our_wads}& $O(n)$ for multi-splay~\cite{multisplay} from empty tree \tabularnewline
\hline 
Sequential  & avoid $(2,1)$ & $O(n)$~\cite{tarjan_sequential, pettie_DS} & $O(n)$~\cite{Fox11, our_wads} & \tabularnewline
\hline 
Traversal  & avoid $(2,3,1)$ & ? & $n2^{\alpha(n)^{O(1)}}$ [* Cor~\ref{cor:trav1}] & $\opt = O(n)$ \tabularnewline
\hline 
\multirow{4}{0.07\paperwidth}{Pattern-avoiding (new)} & avoid $(2,3,1)$ with preprocessing & ? & $O(n)$ [* Cor~\ref{cor:trav2}] & $\opt = O(n)$ \tabularnewline
\cline{2-5} 
 & avoid $(k, \dots, 1)$ ($k$-increasing)  & ? & $n2^{O(k^2)}$ [* Thm~\ref{thm:monotone}] & \tabularnewline
\cline{2-5} 
 & avoid all simple perm.\ of size $> k$ ($k$-decomposable) \\ with preprocessing & ? & $n2^{O(k^{2})}$ [* Thm~\ref{thm:main-noinitial}] & $\opt=O(n\log k)$ \ \ \ [*~Cor~\ref{cor:rgreedy-good}] \tabularnewline
\cline{2-5} 
 & avoid an arbitrary perm.\ of size $k$ & ? & $n2^{\alpha(n)^{O(k)}}$ [* Thm~\ref{thm:pattern-avoidance}] & Best known lower bound is at most $n2^{O(k)}$ [* Thm~\ref{thm:sgreedy}] \tabularnewline
\hline 
\end{tabular}\protect\caption{Easy sequences for BST and best known upper bounds. Results referenced as $[*]$ are new. The symbol ``?'' means that no nontrivial bound is known. In the first rows, $f_{i}$ is the number of times that $i$ is accessed, $x_{t}$ is the element accessed at time $t$, and $\tau_{t}$ is the number of distinct accesses since the last access of $x_{t}$). In the "Sequential", "Traversal", and "Pattern-avoiding" rows it is assumed that the access sequence is a permutation, i.e.\ $m=n$.}

\label{table-results}
\end{table}

Motivated by well-established concepts in combinatorics~\cite{knuth68,Vatter,kitaev}, in this paper  
we initiate the study of the complexity of access sequences in terms of pattern avoidance parameters.  
Our main contributions are two-fold (the precise statement of the results is in \S\,\ref{sec:results}, and the results are summarized in Table~\ref{table-results}). 

\begin{compactenum}[(i)] 
\item We study access sequences parametrized by pattern avoidance and analyze the access time of \greedy in terms of this parameter. 
As a by-product, we almost settle the traversal conjecture for \greedy in two orthogonal directions: (a) \greedy is ``almost linear'' for the traversal conjecture, and (b) if a linear-cost preprocessing is allowed, \greedy is in fact linear.  
This is perhaps the first evidence in support of \greedy as a candidate for dynamic optimality (all previously known properties of \greedy were subsumed by splay). 
These results are derived via an {\em input-revealing property} of \greedy, in the sense that the execution log of \greedy partially reveals the structure of the input sequence.  

\item We study a decomposability property of sequences and show that a wide subclass of the pattern-avoiding class satisfies this property. 
This allows us to (a) identify the core difficulty of designing an offline algorithm for dynamic optimality, (b) obtain a tight bound for the optimum of this input class, and (c) derive a simple proof that Cole's showcase sequence~\cite{finger1} is linear.    
 
\end{compactenum} 

\subsection{Our results }
\label{sec:results}

We first define the pattern avoidance parameters. 
Let $[n] = \{1,\dots,n\}$ be the set of keys. 
We represent a length-$m$ sequence by an $m$-tuple $X = (x_1,\ldots, x_m) \in [n]^m$. 
Throughout the paper, we assume that %By simple arguments (see \Cref{sec:perm-enough}), one may assume that 
$X$ is a permutation, i.e.\ $m=n$ and $x_i \neq x_j$ for all $i \neq j$. This is justified by Theorem~\ref{thm:perm-all}, but we remark that many of our results can be extended to non-permutation access sequences.

Let $S_k$ denote the set of all permutations in $[k]^k$.  
Given $X \in S_n$, we say that $X$ \emph{contains} a permutation pattern $\pi \in S_k$, if there are indices $1 \leq i_1 < i_2 < \cdots i_k \leq n$, such that $(x_{i_1}, \dots, x_{i_k})$ is order-isomorphic to $\pi = (\pi_1, \dots, \pi_k)$. 
Otherwise we say that $X$ \emph{avoids} $\pi$, or that $X$ is $\pi$-\emph{free}. 
Two sequences of the same length are \emph{order-isomorphic} if they have the same pairwise comparisons, e.g.\ $(5,8,2)$ is order-isomorphic to $(2,3,1)$ (Figure~\ref{fig_intro}). As examples, we mention that $(2,1)$-free and $(2,3,1)$-free sequences are exactly the sequential, respectively, preorder access sequences.

For each permutation $X$, a {\em pattern avoidance characteristic} of $X$, denoted by $\Sigma_X$, is the set of all patterns not contained in $X$, i.e.\ $\Sigma_X = \set{\pi \in S_{\ell}:  \mbox{ $\ell \geq 1$ and $X$ is $\pi$-free}}$. 
The {\em pattern avoidance parameter} of $X$, denoted by $k(X)$, is the minimum $k \in \N$ for which $S_k \cap \Sigma_X \neq \emptyset$.  
%$X$ is free of some length-$k$ permutation $\pi$.    

%For each $k \in \N$, let $\Sigma \subseteq \set{\pi \in [\ell]^{\ell}: \ell \leq k}$ be a non-empty collection of patterns of sizes at most $k$. 
For each $k \in \N$, the pattern avoidance class $\cset_k$ contains all sequences $X$ with pattern avoidance parameter $k$.  
This characterization is powerful and universal: $\cset_i \subseteq \cset_{i+1}$ for all $i$, and $\cset_{n+1}$ contains all sequences. %% = S_n$.

\paragraph{Bounds for \greedy in terms of $k$.} 
We study the access cost of \greedy in terms of the pattern avoidance parameter. 
Our most general upper bound of $n 2^{\alpha(n)^{O(k)}}$ holds for any access sequence in $\cset_k$ (Theorem~\ref{thm:pattern-avoidance}). %(in fact, even for non-permutation sequences).
Later, we show a stronger bound of $n2^{O(k^2)}$ for two broad subclasses of $\cset_k$ that have intuitive geometric interpretation: $k$-decomposable sequences (Theorem~\ref{thm:main-noinitial}), which generalize preorder sequences and $k$-increasing sequences (Theorem~\ref{thm:monotone}), which generalize sequential access.   
%further structure is known about $\Sigma_X$, we obtain stronger results:  

\begin{theorem} 
\label{thm:pattern-avoidance}
Let $X \in \cset_k$ be an access sequence of length $n$.  
%that avoids a $k$-by-$k$ pattern $B$.  
The cost of accessing $X$ using \greedy is at most $n 2^{\alpha(n)^{O(k)}}$. 
\end{theorem}

The following corollary simply follows from the fact that all preorder sequences $X$ belong to $\cset_3$ (plugging $k=3$ into the theorem). 

\begin{corollary}[``Almost'' traversal conjecture for \greedy] 
\label{cor:trav1}
The cost of accessing a preorder sequence of length $n$ using \greedy, with arbitrary initial tree is at most $n2^{\alpha(n)^{O(1)}}$.
\end{corollary} 

Next, we show sharper bounds for \greedy when some structural property of $\Sigma_X$ is assumed. 
First, we prove a bound when $\Sigma_X$ contains all ``non-decomposable'' permutations of length at least $k+1$. We show that this property leads to a decomposition that has intuitive geometric meaning, which we describe below.  

Given a permutation $\sigma =(\sigma_1,\ldots, \sigma_{\ell}) \in S_{\ell}$ we call a set $[a,b]$ a \emph{block} of $\sigma$, if $\{\sigma_a, \sigma_{a+1}, \dots, \sigma_b\} = \set{c,c+1,\ldots, d}$ for some integers $c, d \in [\ell]$.
In words, a block corresponds to a contiguous interval that is mapped to a contiguous interval.
We say that a permutation $\sigma \in S_\ell$ is decomposable if there is a block of $\sigma$ of size strictly betweeen $1$ and $\ell$; otherwise, we say that it is \emph{non-decomposable} (a.k.a.\ \emph{simple}\footnote{See \cite{brignall} for a survey of this well-studied concept. The decomposition described here appears in the literature as \emph{substitution-decomposition}.}).
%%For instance a pattern $(2,3,1)$ is decomposable because $[1,2]$ is its block.  

We say that an input sequence $X$ is decomposable into $k$ blocks if there exist disjoint $[a_1,b_1],\ldots,[a_k, b_k]$ such that each $[a_i, b_i]$ is a block for $X$ and $\left( \bigcup_i [a_i, b_i]\right) \cap \N = [n]$. 
The {\em recursive decomposition complexity} of a sequence $X$ is at most $d$ if one can recursively decompose $X$ into subblocks until singletons are obtained, and each step creates at most $d$ blocks (equivalently we say that $X$ is $d$-recursively decomposable, or simply $d$\emph{-decomposable}). See Figure~\ref{fig_intro}.
It is easy to see that preorder sequences are $2$-decomposable.\footnote{The entire set of $2$-decomposable permutations is known as the set of {\em separable permutations}~\cite{separable}.}
 
The following lemma (proof in \Cref{lem:connection-ext}) connects pattern avoidance and recursive decomposition complexity of a sequence.

\begin{lemma}
\label{lem:connection}
Let $X \in S_n$. Then $\Sigma_X$ contains all simple permutations of length at least $k+1$ if and only if $X$ is $k$-decomposable. 
\end{lemma}  

This lemma implies in particular that if $X$ is $k$-decomposable, then $X \in \cset_{k+1}$ and \Cref{thm:pattern-avoidance} can be applied. The following theorem gives a sharper bound, when a preprocessing (independent of $X$) is performed. 
%The exact details of the preprocessing are given in \S\,\ref{sec:preproc}. 
This yields another relaxed version of the traversal conjecture. 
%Now we show the upper bound for \greedy in terms of the recursive decomposition complexity. 

\begin{theorem}
\label{thm:main-noinitial}  
Let $X$ be a $k$-decomposable access sequence.  
The cost of accessing $X$ using \greedy, with a linear-cost preprocessing, is at most $n 2^{O(k^2)}$. 
\end{theorem} 

\begin{corollary}[Traversal conjecture for \greedy with preprocessing]
\label{cor:trav2}
The cost of accessing a preorder sequence using \greedy, with preprocessing, is linear. 
\end{corollary}

Our preprocessing step is done in the geometric view (explained in \S\,\ref{sec:prelim-short}) in order to ``hide'' the initial tree $T$ from \greedy. In the corresponding tree-view, our algorithm preprocesses the initial tree $T$, turning it into another tree $T'$, independent of $X$, and then starts accessing the keys in $X$ using the tree-view variant of \greedy (as defined in~\cite{DHIKP09}). We mention that DHIKP~\cite{DHIKP09} define \greedy without initial tree, and thus their algorithm (implicitly) performs this preprocessing.

As another natural special case, we look at $\Sigma_X \supseteq \{(k,\dots,1)$\}, for some value $k$. In words, $X$ has no decreasing subsequence of length $k$ (alternatively, $X$ can be partitioned into at most $k-1$ increasing subsequences\footnote{This equivalence is proven in Lemma~\ref{lem:decomposition_monotone}.}). Observe that the $k=2$ case corresponds to sequential access. 
We call this the $k$-\emph{increasing} property of $X$. Similarly, if $\Sigma_X \supseteq \{(1,\dots,k)$\}, we say that $X$ is $k$-\emph{decreasing}. 
We obtain the following generalization of the sequential access theorem (this result holds for arbitrary initial tree). 
\begin{theorem}
\label{thm:monotone}  
Let $X$ be a $k$-increasing or $k$-decreasing sequence. The cost of accessing $X$ using \greedy is at most $n 2^{O(k^2)}$. 
\end{theorem} 
 
All the aforementioned results are obtained via the {\em input-revealing} property of \greedy (the ideas are sketched in \S\,\ref{sec:tech}).

\paragraph{Bound for signed {\sc Greedy} in terms of $k$.} We further show the applicability of the input-revealing technique.  
Signed \greedy~\cite{DHIKP09} (denoted \sgreedy) is an algorithm similar to \greedy. It does not always produce a feasible BST solution, but its cost serves as a lower bound for the optimum. Let us denote by $\cA(X)$ and $\opt(X)$ the cost of an algorithm $\cA$ on $X$, respectively the optimum cost of $X$. Then we have $\cA(X) \geq \opt(X)$, and $\opt(X) = \Omega(\textsc{SGreedy}(X))$.

\begin{conjecture}[Wilber~\cite{wilber}, using~\cite{DHIKP09}]
\label{conj:sgreedy}  
$\opt(X) = \Theta(\textsc{SGreedy}(X))$ for all input sequences $X$. 
\end{conjecture} 
 
We show that our techniques can be used to upper bound the cost of \sgreedy.  

\begin{theorem}
\label{thm:sgreedy}
The cost of \sgreedy on an access sequence $X \in \cset_k$ is at most $n 2^{O(k)}$. 

\end{theorem}

\begin{corollary} 
\label{cor:lower bound tight}
Let $\aset$ be an algorithm. If Conjecture~\ref{conj:sgreedy} is true, then one of the following must hold:  
\begin{compactenum}[(i)]
\item The cost of accessing $X$ using $\cA$ is at most $n 2^{O(k)}$, for all $X \in \cset_k$, for all $k$. 
In particular, the traversal conjecture holds for $\cA$.  
\item $\cA$ is not dynamically optimal.  
\end{compactenum} 
\end{corollary}

\paragraph{Decomposability.}
Next, we prove a general theorem on the complexity of sequences with respect to the recursive decomposition.
Let $X \in S_n$ be an arbitrary access sequence. We represent the process of decomposing $X$ into blocks until singletons are obtained by a tree $\tset$, where each node $v$ in $\tset$ is associated with a sequence $X_v$ (see Figure~\ref{fig_intro}).

\begin{theorem}[Decomposition theorem, informal]
\label{thm:intro-decomposition} 
$\opt(X) \leq \sum_{v \in \tset} \textsc{Greedy}(X_v) + O(n) $.
\end{theorem}

\begin{corollary} 
\label{cor:rgreedy-good}
Let $X$ be a $k$-decomposable sequence of length $n$ (with $k$ possibly depending on $n$). Then $X$ has optimum cost $\opt(X) \leq O(n \log k)$.
\end{corollary} 

This result is tight for all values of $k$, since for all $k$, there is a $k$-decomposable sequence whose complexity is $\Omega(n \log k)$ (Proof in Appendix~\ref{sec:tight-block}). 
This result gives a broad range of input sequences whose optimum is well below the performance of all known BST algorithms. The class can thus be useful as ``candidate counterexample'' for dynamic optimality. (Since we know $\opt$, we only need to show that no online algorithm can match it.) 
We remark that \greedy asymptotically matches this bound when $k$ is constant (Theorem~\ref{thm:main-noinitial}).

The following observation follows from the fact that the bound in Theorem~\ref{thm:intro-decomposition} is achieved by our proposed offline algorithm, and the fact that $\opt(X) \geq \sum_{v \in \tset} \opt(X_v)$ (see Lemma~\ref{lem:decomp opt lower-bound}).

\begin{corollary}
If \greedy is constant competitive on non-decomposable sequences, then there is an offline $O(1)$-approximate algorithm on any sequence $X$.
\end{corollary} 

This corollary shows that, in some sense, non-decomposable sequences serve as the ``core difficulty'' of proving dynamic optimality, and understanding the behaviour of \greedy on them might be sufficient.

Finally, Cole et al.~\cite{finger1} developed sophisticated technical tools for proving the dynamic finger theorem for splay trees. 
They introduced a ``showcase'' sequence to illustrate the power of their tools, and showed that splay has linear cost on such a sequence. 
Theorem~\ref{thm:intro-decomposition} immediately gives an upper bound on the optimum access cost of this sequence. 

\begin{corollary} [informal] 
Let $X$ be Cole's showcase sequence. Then $\opt(X) = O(n)$. 
\end{corollary}

\begin{figure}[h]
\centering
\includegraphics[scale=1.1]{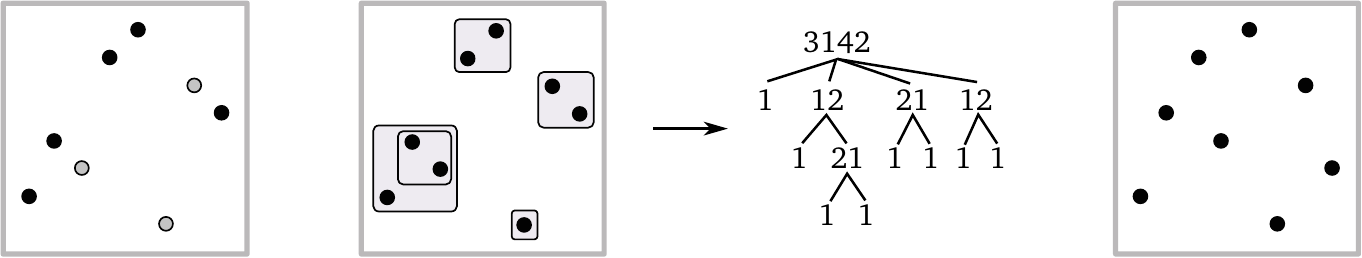}
\caption{From left to right: (i) plot of permutation $\sigma= (6,1,3,2,8,7,4,5)$ containing $(2,1,3)$ (highlighted) and avoiding $(4,3,2,1)$, (ii) permutation $\sigma$, and (iii) its block decomposition tree; permutation $(3,1,4,2)$ at the root is obtained by contracting the four blocks into points, (iv) simple permutation $(6,1,8,4,2,7,3,5)$.}
\label{fig_intro}
\end{figure}

\subsection{Overview of techniques}
\label{sec:tech}
We sketch the main technical ideas of our proofs.
First, we define roughly the execution log matrix of \greedy. 
The log of executing \greedy on $X \in [n]^m$ is an $n$-by-$m$ matrix $G_X$ such that $G_X(a,t) = 1$ if and only if element $a$ is touched by \greedy at time $t$. The cost of this execution is exactly the number of ones in $G_X$, denoted by $w(G_X)$ (we make this claim precise in \S\,\ref{sec:prelim-short}). 

A submatrix of $G_X$ is obtained by removing some rows and columns of $G_X$.
We say that $G_X$ \emph{contains} a matrix $B$ if there is a submatrix $B'$ of $G_X$ such that $B(i,j) \leq B'(i,j)$ for all $i,j$ (in words, any `one' entry in $B$ must also be a one in $B'$).
 
We say that an algorithm $\aset$ is {\em input-revealing} if there is a constant-size matrix $B$ such that any submatrix of an execution matrix of $\aset$ containing $B$ must contain an access point. 
When such a matrix $B$ exists for an algorithm, we say that $B$ is a {\em capture gadget}.  
The existence for \greedy of a capture gadget $B$ allows us to observe that $G_X$ contains the pattern\footnote{$\otimes$ denotes the tensor (Kronecker) product, to be defined later.} $Q \otimes B$ only if the input sequence contains the pattern defined by $Q$. 
So, if we execute \greedy on an input sequence $X \in \cset_k$ that is $Q$-free for $Q \in S_k$, then $G_X$ cannot contain $Q \otimes B$.  
Now, we can use results from forbidden submatrix theory for bounding $w(G_X)$. 
In short, forbidden submatrix theory is a collection of theorems of the following type: 

\begin{theorem} 
Let $M$ and $P$ be $n$-by-$n$ and $k$-by-$k$ matrices respectively, such that $k \leq n$ (mostly $k$ should be thought of as constant).  
If $M$ does not contain $P$, then $w(M) \leq n f_P(n,k)$.
\end{theorem} 

The function $f_P$ depends on the matrix $P$ that is forbidden in $M$.  
For instance, if $P$ contains at most one $1$-entry per row and column, then $f_P(n,k)$ is only a function of $k$.  
If $P$ contains at most one $1$-entry per column but possibly multiple $1$s per row, then $f_P(n,k) = 2^{\alpha(n)^{O(k)}}$. 
The tightness of the bounds we obtain depends on the structure of the capture gadget $B$. 
This is the reason we have a weaker bound for any $X \in \cset_k$ and stronger bounds when $X$ is more restricted.    
 
We also develop a different set of tools to study the intrinsic complexity of decomposable sequences. Here, in order to compute bounds on \opt, we decompose the execution matrix $M_X$ of an algorithm in a way that mirrors the recursive decomposition of the input sequence $X$. Unfortunately, the behaviour of \greedy is too capricious to afford such a clean decomposition. We therefore modify \greedy to reduce the interference between blocks in the input sequence. The modified algorithm might perform more work locally than \greedy, but globally it is more ``robust''. It is however, an offline algorithm; whether there exists an online algorithm competitive with our robust \greedy, is an interesting open question.

\paragraph{Related work.}
For an overview of results related to dynamic optimality, we refer to the survey of Iacono~\cite{in_pursuit}. 
Besides Tango trees, other $O(\log\log n)$-competitive BST algorithms have been proposed, similarly using Wilber's first bound (multi-splay~\cite{multisplay} and chain-splay~\cite{chain_splay}). Using~\cite{DemaineILO13}, all the known bounds can be simultaneously achieved by combining several BST algorithms. Our study of the execution log of BST algorithms makes use of the geometric view of BST introduced by DHIKP~\cite{DHIKP09}. 

Pattern avoidance in sequences has a large literature in combinatorics~\cite{Vatter,kitaev}, as well as in computer science~\cite{BrunerL13,knuth68, tarjan_sorting, pratt_queues, separable}. The theory of forbidden submatrices can be traced back to a 1951 problem of Zarankiewicz, and its systematic study has been initiated by F\"{u}redi~\cite{furedi}, and Bienstock and Gy\H{o}ri~\cite{bienstock_gyori}. Our use of this theory is inspired by work of Pettie~\cite{pettie_DS, deque_Pet08}, who applied forbidden pattern techniques to prove bounds about data structures (in particular he proves that splay trees achieve $O(n\alpha^{\ast}(n))$ cost on deque-sequences, and he reproves the sequential access theorem for splay trees). There are marked differences in our use of the techniques compared with the work of Pettie. In particular, we look for patterns directly in the execution log of a BST algorithm, without any additional encoding step. Furthermore, instead of using fixed forbidden patterns, we make use of patterns that depend on the input. 

\paragraph{Organization of the paper.} In \S\,\ref{sec:prelim-short} we give a short description of the geometric view of the BST model (following DHIKP~\cite{DHIKP09}), and we introduce concepts used in the subsequent sections. In \S\,\ref{sec:input-rev} we study the input-revealing property of \greedy, and prove our bounds for pattern-avoiding sequences. In \S\,\ref{sec:dec-short} we introduce a robust \greedy, state the decomposition theorem and its applications. Finally, \S\,\ref{sec:conclusions} is reserved for open questions. We defer the proofs of several results, as well as some auxilliary material to the Appendix.

\section{Preliminaries}
\label{sec:prelim-short}

The main object of study in this paper are $\{0,1\}$-matrices (corresponding to the execution matrix of an algorithm). 
For geometric reasons, we write matrix entries bottom-to-top, i.e.\ the first row is the bottom-most. Thus, $M(i,j)$ is the value in the $i$-th column (left-to-right) and the $j$-th row (bottom-to-top) of $M$. We simultaneously view $M$ as a set of integral points in the plane. For a point $p=(x,y)$, we say that $p \in M$, if $M(x,y) = 1$.
We denote by $p.x$ and $p.y$ the $x$- and $y$-coordinates of point $p$ respectively.

Given a point set (matrix) $M$, let $p$ and $q$ be points not necessarily
in $M$. Let $\square_{pq}$ be the closed rectangle whose two corners
are $p$ and $q$. Rectangle $\square_{pq}$ is \emph{empty} if
$M \cap \square_{pq} \setminus \{p,q\} = \emptyset$. Rectangle $\square_{pq}$
is \emph{satisfied} if $p,q$ have the same $x$ or $y$ coordinate,
or there is another point $r\in M \cap \square_{pq} \setminus\{p,q\}$.
A point set $M$ is \emph{(arboreally) satisfied} if  
%if for any two points
%$p,q\in M$, 
$\square_{pq}$ is satisfied for any two points $p,q \in M$. %%%The concept of satisfied point sets was introduced in~\cite{DHIKP09}.

\paragraph{Geometric setting for online BST.}

%%\footnote{Should I use the term ``online algorithm'' or ``online BST''?%
%%}

We describe the geometric model for analyzing online BST algorithms (or simply online BST). 
This model is essentially the same as~\cite{DHIKP09}, except for the fact that our model allows cost analysis of  BST algorithms that start accessing an input sequence from arbitrary initial tree.  

For any access sequence $X\in [n]^m$, we also refer to $X$ as an
$n$-by-$m$ matrix, called \emph{access matrix}, where there is a
point $(x,t)\in X$ iff an element $x$ is accessed at time $t$.
We call $(x,t)$ the access point (or input point) at time $t$.
We also refer to the $x$-th column of $X$ as \emph{element} $x$,
and to the $t$-th row of $X$ as \emph{time} $t$. 

For any tree $T$ of depth $d$, let $d(x)$ be the depth of element
$x$ in $T$. We also refer to $T$ as an $n$-by-$d$ matrix, called
\emph{initial tree matrix}, where in each column $x$, there is a
stack of points $(x,1),\dots,(x,d-d(x)+1)\in T$. See \Cref{fig:geo init tree} for an illustration. 
Observe that matrix $T$ is satisfied.

An \emph{online geometric} BST $\cA$ accessing sequence $X$ starting
with initial tree $T$ works as follows. Given a matrix
$\bigl[ \begin{smallmatrix} 
  X\\
  T 
\end{smallmatrix} \bigr]$ 
as input, $\cA$ outputs a satisfied matrix
$\bigl[ \begin{smallmatrix} 
  \cA_{T}(X)\\
  T 
\end{smallmatrix} \bigr]$,
where $\cA_{T}(X)\supseteq X$ is an $n$-by-$m$ matrix called
the \emph{touch matrix} of $\cA$ on $X$ with initial tree $T$.
More precisely, $\aset$ proceeds in an online fashion: At time $t=1,\ldots, n$, the algorithm $\aset$ has access to $T$ as well as to the first $t$ rows of matrix $X$, and it must decide on the $t$-th row of the output $\aset_T(X)$ which cannot be changed afterwards. 
Algorithm $\aset$ is {\em feasible} if  
$\bigl[ \begin{smallmatrix} 
\cA_{T}(X)\\
  T 
\end{smallmatrix} \bigr]$ is satisfied.  

The matrix $\cA_{T}(X)$ can be thought of as the execution log of $\cA$, and points in $\cA_{T}(X)$ are called \emph{touch points}; the \emph{cost} of
$\cA$ is simply $w(\aset_T(X))$, i.e.\ the number of points (ones) in $\cA_{T}(X)$.

An online algorithm with {\em preprocessing} is an online algorithm that is allowed to perform differently at time $t=0$.  
The linear-cost preprocessing used in this paper, as well as in DHIKP~\cite{DHIKP09}, is to put all elements into a \emph{split tree} (see~\cite{DHIKP09} for the details of this data structure). The consequences of this preprocessing are that (i) in the geometric view the initial tree matrix can be removed, i.e.\ an algorithm $\aset$ only gets $X$ as input and outputs $\aset(X)$, and (ii) in the tree-view, an input initial tree $T$ is turned into another tree $T'$ independent of the input $X$, and then keys in $X$ are accessed using the tree-view variant of \greedy. 

DHIKP~\cite{DHIKP09} % 
showed\footnote{with some minor additional argument about initial trees from \Cref{sec:initial-tree}.} that online \emph{geometric} BSTs are essentially equivalent
to online BSTs (in tree-view). In particular, given any online geometric BST 
$\cA$, there is an online BST $\cA'$ where the cost of running
$\cA'$ on any sequence $X$ with any initial tree $T$ is a constant
factor away from $w(\cA_{T}(X))$.
Therefore, in this paper, we focus on online geometric BSTs and refer to them from now on simply as online BSTs. 

Unless explicitly stated otherwise, we assume that the access
sequence $X$ is a permutation, i.e.\ each element is accessed only
once, and thus $X$ and $\cA_T(X)$ are $n$-by-$n$ matrices.
This is justified by the following theorem (proof in Appendix~\ref{sec:perm-enough}). 
\begin{theorem}
	\label{thm:perm-all}
	Assume that there is a $c$-competitive geometric BST $\cA^{p}$ for
	all permutations. Then there is a $4c$-competitive geometric
	BST $\cA$ for all access sequences.\footnote{The statement holds unconditionally for offline algorithms, and under a mild assumption if we require $\cA^{p}$ and $\cA$ to be online.}\end{theorem}

\paragraph{\greedy and signed \greedy.}

Next, we describe the algorithms \greedy and signed \greedy~\cite{DHIKP09} (denoted as \sgreedy). 

Consider the execution of an online BST $\cA$ on sequence $X$ with initial
tree $T$ just before time $t$. 
Let $O_{t-1}$ denote the output at time $t-1$ (i.e.\ the initial tree $T$ and the first $t-1$ rows of $\cA_{T}(X)$). 
Let $\tau(b,t-1)=\max\left\{ i\mid(b,i)\in O_{t-1}\right\} $ for any
$b\in[n]$. In words, $\tau(b,t-1)$ is the last time \emph{strictly before} $t$ when $b$ is touched, otherwise it is a non-positive number defined by the initial tree.
Let $(a,t)\in X$ be the access point at time $t$. We define
$\stair_{t}(a)$ as the set of elements $b\in[n]$ for which $\square_{(a,t),(b,\tau(b,t-1))}$
is not satisfied.
\greedy touches element $x \in [n]$ if and only if $x \in \stair_t(a)$.   
See \Cref{fig:geo init tree} for an illustration of stairs.

\begin{figure}[h]
\begin{center}  
\includegraphics[width=0.65\textwidth]{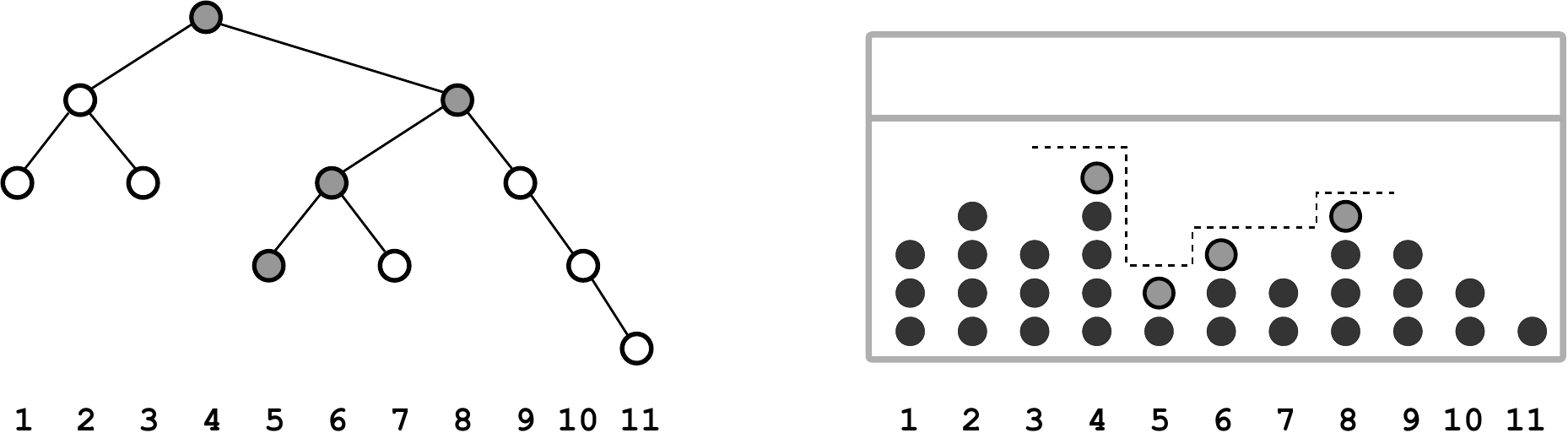}
\end{center}
\caption{\label{fig:geo init tree}Tree-view and geometric view of initial tree. Search path and stair (at time 1) of element 5.}
\end{figure}

We define $\stairleft_{t}(a)=\{b\in\stair_{t}(a)\mid b<a\}$
and $\stairright_{t}(a)=\{b\in\stair_{t}(a)\mid b>a\}$. The algorithms \greedyleft
and \greedyright are analogous to \greedy, but they only touch $\stairleft_{t}(a)$, respectively 
$\stairright_{t}(a)$ for every time $t$; \sgreedy simply refers to the union of the \greedyleft and \greedyright outputs, which is not necessarily satisfied. It is known~\cite{DHIKP09, Harmon} that \sgreedy constant-approximates the \emph{independent rectangle bound} which subsumes the two previously known lower bounds on \opt, called Wilber's first and second bounds~\cite{wilber}.
\begin{theorem} [\!\cite{wilber, DHIKP09, Harmon}]
For any BST $\cA$ (online or offline) that runs on sequence $X$ with arbitrary initial tree $T$, $w(\cA_{T}(X))=\Omega(\textsc{SGreedy}(X))$. %
\end{theorem}

A nice property of \greedy and \sgreedy is that an element $x$ can become ``hidden'' in some range in the sense that 
if we do not access in $x$'s hidden range, then $x$ will not be touched. Next, we formally define this concept, and list some cases when elements become hidden during the execution of \greedy and \sgreedy. The concept is illustrated in Figure~\ref{fig:hidden-show}.

\begin{figure}[h]
\begin{center}  
\includegraphics[width=0.35\textwidth]{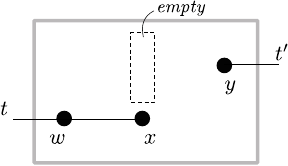}
\end{center}
\caption{\label{fig:hidden-show} After time $t$, element $x$ is hidden in $(w,n]$ for \greedy, and in $(w,x]$ for \greedyright. 
After time $t'$, element $x$ is hidden in $(w,y)$ for \greedy; it is easy to verify that any access outside of $[w,y]$ will never touch $x$.}
\end{figure}

\begin{definition}[Hidden]\label{def:hidden}
	For any algorithm $\cA$, an element $x$ is \emph{hidden in} range $[w,y]$ after $t$ if,
	given that there is no access point $p \in [w,y]\times(t,t']$ for any $t'>t$, then
	 $x$ will not be touched by $\cA$ in time $(t,t']$.
\end{definition}
\begin{lemma} [Proof in Appendix~\ref{sec:wing-hidden}]
\label{prop:hidden}Let $x$ be some element.
\begin{compactenum}[(i)]
\item If there is an element $w<x$ where $\tau(w,t)\ge\tau(x,t)$, then
$x$ is hidden in $(w,n]$ and $(w,x]$ after $t$ for \greedy and \greedyright respectively.
\item If there is an element $y>x$ where $\tau(y,t)\ge\tau(x,t)$, then
$x$ is hidden in $[1,y)$ and $[x,y)$ after $t$ for \greedy and \greedyleft respectively.
\item If there are elements $w,y$ where $w<x<y$, and $\tau(w,t),\tau(y,t)\ge\tau(x,t)$, then
$x$ is hidden in $(w,y)$ after $t$ for \greedy.
\end{compactenum}
\end{lemma}

\paragraph{Forbidden submatrix theory.}

Let $M$ be an $n$-by-$m$ matrix. Given $I \subseteq [n]$ and $J \subseteq [m]$, a \emph{submatrix} $M|_{I,J}$ is obtained by removing columns in $[n] \setminus I$ and rows in $[m] \setminus J$ from $M$.

\begin{definition}
[Forbidden submatrix]Given a matrix $M$ and another matrix
$P$ of size $c\times d$, we say that $M$ \emph{contains}
$P$ if there is a $c\times d$ submatrix $M'$ of $M$ such that
$M'(i,j)=1$ whenever $P(i,j)=1$. We say that $M$ \emph{avoids} (forbids) $P$ if $M$
does not contain $P$. 
\end{definition}

We often call an avoided (forbidden) matrix a ``pattern''. 
For any permutation $\pi = (\pi_1,\dots,\pi_k) \in S_k$, we also refer to $\pi$ as a $k$-by-$k$ matrix where $(\pi_i,i) \in \pi$ for all $1\le i\le k$. % Observe that if an access sequence $S$ avoids a permutation $P$, then the access matrix $S$ is $M_P$-free. 
Permutation matrices have a single one-entry in every row and in every column. Relaxing this condition slightly, we call a matrix that has a single one in every column a \emph{light} matrix. We make use of the following results.
\begin{lemma}
\label{lem:forbidden}
	Let $M$ and $P$ be $n$-by-$n$ and $k$-by-$k$ matrices such that $M$ avoids $P$.  
	\begin{compactenum}[(i)]
		\item (Marcus, Tardos~\cite{marcus_tardos}, Fox~\cite{jfox})\ \  If $P$ is a permutation, the number
		of ones in $M$ is at most $n2^{O(k)}$.
		\item (Nivasch~\cite{nivasch}, Pettie~\cite{pettie_nonlinear})\ \ If $P$ is light, the number of ones in $M$ is at most $n2^{\alpha(n)^{O(k)}}$.
	\end{compactenum}
\end{lemma}

We define the \emph{tensor product} between matrices as follows: $M \otimes G$ is the matrix obtained by replacing each one-entry of $M$ by a copy of $G$, and replacing each zero-entry of $M$ by an all-zero matrix equal in size with $G$. Suppose that $M$ contains $P$. A \emph{bounding
box} $B$ of $P$ is a minimal submatrix $M|_{I,J}$ such that $P$ is contained in $B$, and $I$ and $J$ are sets of consecutive columns, resp.\ rows. We denote $\min(I)$, $\max(I)$, $\min(J)$, $\max(J)$, as $B.\xmin$, $B.\xmax$, $B.\ymin$, $B.\ymax$, respectively.

\section{\greedy is input-revealing}
\label{sec:input-rev}
In this section, we prove \Cref{thm:pattern-avoidance,thm:main-noinitial,thm:monotone,thm:sgreedy}.
We refer to \S\,\ref{sec:results} for definitions of pattern-avoiding, $k$-decomposable and $k$-increasing permutations.

To bound the cost of algorithm $\cA$, we show that $\aset_T(X)$ (or $\aset(X)$ if preprocessing is allowed) avoids some pattern.  
In the following, $\inc_k$, $\dec_k$, and $\alt_k$ are permutation matrices of size $k$, and $\capture$ is a light matrix. These matrices will be defined later.
\begin{lemma}
	\label{lem:touch avoid}Let $X$ be an access sequence.
\begin{compactenum}[(i)]
		\item If $X$ avoids a permutation $P$ of size $k$, then for any initial
		tree $T$, $\textsc{GreedyRight}_{T}(X)$, $\textsc{GreedyLeft}_{T}(X)$ and $\Greedy_{T}(X)$
		avoid $P\otimes(1,2)$, $P\otimes(2,1)$ and $P\otimes\capture$
		respectively.
		\item If $X$ is $k$-increasing ($k$-decreasing), then for any initial
		tree $T$, $\Greedy_{T}(X)$ avoids $(k,\dots,1)\otimes\dec_{k+1}$
		($(1,\dots,k)\otimes\inc_{k+1}$).
		\item If $X$ is $k$-decomposable, then $\Greedy(X)$ avoids $P\otimes\alt_{k+4}$
		for all simple permutations $P$ of size at least $k+1$.
	\end{compactenum}
\end{lemma}

The application of Lemma~\ref{lem:forbidden} is sufficient to complete all the proofs, together with the observation that, for any permutation matrix $P$ of size $k \times k$, if $G$ is a permutation (resp.\ light) matrix of size $\ell \times \ell$, then $P \otimes G$ is a permutation (resp.\ light) matrix of size $k \ell \times k \ell$.

The rest of this section is devoted to the proof of \Cref{lem:touch avoid}.
We need the definition of the \emph{input-revealing gadget} which is a pattern
whose existence in the touch matrix (where access points are indistinguishable from touch
points) gives some information about the locations of access points.
\begin{definition}
	[Input-revealing gadgets]\label{def:gadget}Let $X$ be an access matrix, $\cA$ be
	an algorithm and $\cA(X)$ be the touch matrix of $\cA$ running on
	$X$. 
	\begin{compactenum}[--]
		\item {\bf Capture Gadget.} A pattern $P$ is a \emph{capture gadget }for algorithm
		$\cA$ if, for any bounding box $B$ of $P$ in $\cA(X)$, there is
		an access point $p\in B$.
		\item {\bf Monotone Gadget.} A pattern $P$ is a \emph{$k$-increasing ($k$-decreasing)
			gadget }for algorithm $\cA$ if, for any bounding box $B$ of $P$
		in $\cA(X)$, either (i) there is an access point $p\in B$ or, (ii)
		there are $k$ access points $p_{1},\dots,p_{k}$ such that $B.\ymin \le p_{1}.y<\dots<p_{k}.y\le B.\ymax$
		and $p_{1}.x<\dots<p_{k}.x<B.\xmin$ \emph{(}$p_{1}.x>\dots>p_{k}.x>B.\xmax$\emph{)}.
		In words, $p_{1},\dots,p_{k}$ form the pattern $(1,\dots,k)$ (resp.\ $(k,\dots,1)$).
		\item {\bf Alternating Gadget.} A pattern $P$ is a \emph{$k$-alternating gadget
		}for algorithm $\cA$ if, for any bounding box $B$ of $P$ in $\cA(X)$,
		either (i) there is an access point $p\in B$ or, (ii) there are $k$
		access points $p_{1},\dots,p_{k}$ such that $B.\ymin \le p_{1}.y<\dots<p_{k}.y\le B.\ymax$
		and $B.\xmax<p_{i}.x$ for all odd $i$ and $p_{i}.x<B.\xmin$ for all
		even $i$.
	\end{compactenum}
\end{definition}

Notice that if we had a permutation capture gadget $G$ for algorithm $\aset$, we would be done: If $X$ avoids pattern $Q$, then $\aset_T(X)$ avoids $Q \otimes G$ and forbidden submatrix theory applies. With arbitrary initial tree $T$, it is impossible to construct a permutation capture gadget for \greedy even for $2$-decomposable sequences (see Appendix~\ref{sec:bad-example2}). Instead, we define a \emph{light} capture gadget for \greedy for any sequence $X$ and any initial tree $T$. Additionally, we construct permutation monotone and alternating gadgets. After a preprocessing step, we can use the alternating gadget to construct a permutation capture gadget for $k$-decomposable sequences.

\begin{lemma}
\label{lem:input-revealing}
\begin{compactenum}[(i)]
		\item $(12)$ is a capture gadget for \greedyright and $(21)$ is a capture
		gadget for \greedyleft.
		\item $\capture$ is a capture gadget for \greedy where $\capture=\left(\begin{smallmatrix}
		& \bullet\\
		\bullet &  & \bullet
		\end{smallmatrix}\right)$.
		\item $\inc_{k+1}$ ($\dec_{k+1})$ is a $k$-increasing ($k$-decreasing)
		gadget for \greedy where $\inc_{k}=(k,\dots,1)$ and \textup{$\dec_{k}=(1,2,\dots,k)$}. 
		\item $\dec_{k+1}$ ($\inc_{k+1})$ is a capture gadget for \greedy that
		runs on $k$-increasing $(k$-decreasing$)$ sequence.
		\item $\alt_{k+1}$ is a $k$-alternating gadget for \greedy where $\alt_{k}=(\left\lfloor (k+1)/2\right\rfloor,k,1,k-1,2,\dots)$.%%\footnote{We mention that the pattern $\alt_k$ is not the
%%			only possibility for an alternating gadget.} 
		\item $\alt_{k+4}$ is a capture gadget for \greedy that runs on $k$-decomposable
		sequence without initial tree.
	\end{compactenum}
\end{lemma}
For example, 
$
\alt_{5}=\left(\begin{smallmatrix}
& \bullet\\
&  &  & \bullet\\
\bullet\\
&  &  &  & \bullet\\
&  & \bullet
\end{smallmatrix}\right)
$. 
We emphasize that only result (vi) requires a preprocessing. 

Given Lemma~\ref{lem:input-revealing}, and the fact that $k$-increasing, $k$-decreasing, and $k$-decomposable sequences avoid 
$(k,\dots,1)$, $(1,\dots,k)$, resp.\ all simple permutations of size at least $k+1$, \Cref{lem:touch avoid} follows immediately.

\begin{proof}
We only show the proofs for (i) and (ii) in order to convey the flavor of the arguments. 
The other proofs are more involved. See \Cref{sub:k inc seq} for proofs of (iii),(iv), and \Cref{sub:k dec seq} for proofs of (v),(vi).

	(i) We only prove for \GreedyR. Let $a,b$ be touch points
	in $\GreedyR_{T}(X)$ such that $a.y<b.y$ and form $(1,2)$ with
	a bounding box $B$. As $a.y=\tau(a.x,a.y)\ge\tau(b.x,a.y)$, by \Cref{prop:hidden}~(i),
	$b.x$ is hidden in $(a.x,b.x]$ after $a.y$. 
	Suppose	that there is no access point in $B = \square_{ab}$, then $b$ cannot be touched, which is a contradiction. 
	
	(ii) Let $a,b,c$ be touch points in $\Greedy_{T}(X)$ such that $a.x<b.x<c.x$
	and $a,b,c$ form $\capture$ with bounding box $B$. As $\tau(a.x,a.y),\tau(c.x,a.y)\ge\tau(b.x,a.y)$, by \Cref{prop:hidden}~(iii), 
	$b.x$ is hidden in $(a.x,c.x)$ after $a.y$.
	Suppose that there is no access point in $B$, then $b$ cannot be touched, which is a contradiction. 
\iffalse 	
	(v),(vi): This proofs are quite involved and we refer to \Cref{sec:block_linear}
	for the complete proofs. 
	The idea is of the proof is the following.
	Suppose that $\Greedy(X)$ contains $\alt_{k+4}$ with bounding box $B$. 
	Since there	is no initial tree, it can be shown that there exists an access	point $p_{0}$ below $B$. 
	Suppose that $B$ contains no access point, as $\alt_{k+4}$ is an alternating gadget, 
	there are access points $p_{1},\dots,p_{k+4}$  that
	``alternate side'' across $B$, and, in particular, across $p_0$ below.
	By $k$-decomposability, we an show that this is impossible.
	So $B$ must contain an access point, and hence $\alt_{k+4}$ is a capture gadget for \greedy running
	on $k$-decomposable sequence.
\fi 
\end{proof}

\section{Decomposition theorem and applications}
\label{sec:dec-short}
In this section, we take a deeper look at decomposable permutations and investigate whether \greedy is ``decomposable'', in the sense that one can bound the cost of \greedy on input sequences by summing over local costs on blocks. Proving such a property for \greedy seems to be prohibitively difficult (we illustrate this in Appendix~\ref{sec:bad-example1}). Therefore, we propose a more ``robust'' but offline variant of \greedy which turns out to be decomposable. We briefly describe the new algorithm, and state the main theorems and three applications, leaving the details to Appendix~\ref{sec:decomposition}. 

\vspace{-0.06in}
\paragraph*{Recursive decomposition.}

We refer the reader to \S\,\ref{sec:results} for the definitions of blocks and recursive decomposition. Let $S$ be a $k$-decomposable permutation access sequence, and let $\tset$ be its (not necessarily unique) decomposition tree (of arity at most $k$).
Let $P$ be a block in the decomposition tree $\tset$ of $S$. Consider the subblocks $P_1, P_2, \dots, P_k$ of $P$, and let $\tilde{P}$ denote the unique permutation of size $[k]$ that is order-isomorphic to a sequence of arbitrary representative elements from the blocks $P_1, P_2, \dots, P_k$. We denote $P = \tilde{P}[P_1, P_2, \dots, P_k]$ a \emph{deflation} of $P$, and we call $\tilde{P}$ the \emph{skeleton} of $P$. Visually, $\tilde{P}$ is obtained from $P$ by contracting each block $P_i$ to a single point. We associate with each block $P$ (i.e.\ each node in $\tset$) both the corresponding rectangular region (in fact, a square), and the corresponding skeleton $\tilde{P}$ (Figure~\ref{fig_intro}). We denote the set of leaves and non-leaves of the tree $\tset$ by $L(\tset)$ and $N(\tset)$ respectively. %%%We remark that $L(\tset)$ contains only simple permutations.

Consider the $k$ subblocks $P_1,\ldots,P_k$ of a block $P$ in this decomposition, numbered by increasing $y$-coordinate. We partition the rectangular region of $P$ into $k^2$ rectangular regions, whose boundaries are aligned with those of $P_1,\ldots,P_k$. We index these regions as $R(i,j)$ with $i$ going left to right and $j$ going bottom to top. In particular, the region $R(i,j)$ is aligned (same $y$-coordinates) with $P_j$.

\vspace{-0.06in}
\paragraph*{A robust version of \greedy.}
The main difficulty in analyzing the cost of \greedy on decomposable access sequences is the ``interference'' between different blocks of the decomposition. Our \agreedy algorithm is an extension of \greedy. Locally it touches at least those points that \greedy would touch; additionally it may touch other points, depending on the recursive decomposition of the access sequence (this a priori knowledge of the decomposition tree is what makes \agreedy an offline algorithm). The purpose of touching extra points is to ``seal off'' the blocks, limiting their effect to the outside.

We now define the concept of \emph{topwing}. The \textit{topwing} of a rectangular region $R$ at a certain time in the execution of \agreedy contains a subset of the points in $R$ (at most one in each column). A point is in the topwing if it forms an empty rectangle with the top left or the top right corner of $R$ (this definition includes the topmost touch point of the leftmost and rightmost columns of $R$, if such points exist in $R$). %, and 2) the topmost points in the left and right columns of $R$. 
When accessing an element at time $t$, \agreedy touches the projections to the timeline $t$ of the topwings of all regions $R(i,j)$ aligned with a block $P_j$ ending at time $t$ (this includes the region of $P_j$ itself). Observe that multiple nested blocks might end at the same time, in this case we process them in increasing order of their size. %%By touching these topwings, the interference between regions is limited, which allows us to decompose the cost of \agreedy. 
See Figure~\ref{fig_rgreedy} for an illustration of \agreedy.

\begin{figure}[h]
\centering
\includegraphics[scale=1.4]{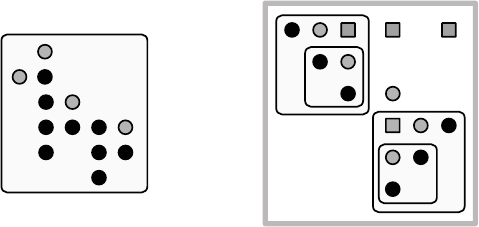}
\caption{Illustration of \agreedy. Left: region with points and \textit{topwing} highlighted. Right: a sample run of \agreedy on a $2$-decomposable input. Access points are dark circles, points touched by \greedy are gray circles, points touched by the \agreedy augmentation are gray squares. The access point in the topmost line completes the black square and also the enclosing gray square. The topwing of the black square consists of the black circle at position (1,6) and the gray circle at (3,5). Therefore the augmentation of the black square adds the gray square at (3,6). The topwing of $R(2,2)$ consists of the gray circle at (4,4) and hence the gray square at (4,6) is added. The topwing of the gray square consists of points (1,6) and (6,3); the point (4,4) does not belong to the topwing because (4,6) has already been added. Therefore, the gray square at (6,6) is added.}
\label{fig_rgreedy}
\end{figure}

\begin{theorem}[Decomposition theorem, Theorem~\ref{lem: main}]
\label{thm:main-sh} 

For any decomposition tree $\tset$ of a permutation $P$, the total cost of \agreedy is bounded as 
\vspace{-0.05in}
\[ \textsc{RGreedy}(P) \leq 4 \cdot \sum_{\tilde P \in N(\tset) }  \textsc{Greedy}(\tilde P) + \sum_{P \in L(\tset) } \textsc{Greedy}(P)+ 3n.\]
\end{theorem}

\paragraph{Application 1: Cole et al.'s ``showcase'' sequences.} 
%In Cole et al.~\cite{finger1}, a showcase sequence is defined in order to introduce technical tools for proving the dynamic finger theorem. 
This sequence can be defined as a permutation $P = \tilde P[P_1,\ldots, P_{n/\log n}]$ where each $P_j$ is the permutation $(1,\ldots, \log n)$, i.e.\ a sequential access. 
It is known that \greedy is linear on sequential access, so $\textsc{Greedy}(P_j)  = O( \log n)$.  
%sequence for which \greedy has linear access cost, i.e.\ $|P_j| = \log n$, and $\textsc{Greedy}(P_j) = O(\log n)$. 
Now, by applying Theorem~\ref{thm:main-sh}, we can say that the cost of \agreedy on $P$ is $\textsc{RGreedy}(P) \leq 4 \cdot \textsc{Greedy}(\tilde P) + \sum_{j} \textsc{Greedy}(P_j) + O(n) \leq O(\frac{n}{\log n} \cdot \log (\frac{n}{\log n})) + \frac{n}{\log n}O(\log n) + O(n) = O(n)$. 
This implies the linear optima of this class of sequences.  

\paragraph{Application 2: Simple permutations as the ``core'' difficulty of BST inputs.}  

We prove that, if \greedy is $c$-competitive on simple permutations, then \agreedy is $O(c)$-approximate for all permutations. 

To prove this, we consider a block decomposition tree $\tset$ in which the permutation corresponding to every node is simple. Using Theorem~\ref{thm:main-sh} and the hypothesis that \greedy is $c$-competitive for simple permutations:

\[\textsc{RGreedy}(P) \leq 4c \cdot \sum_{\tilde P \in  N(\tset)} \opt(\tilde P) + c \cdot \sum_{P \in L(\tset)} \opt(P) +3n.\]

Now we decompose $\opt$ into subproblems the same way as done for \agreedy. We prove in \Cref{lem:decomp opt lower-bound} that $\opt(P) \geq \sum_{P \in \tset } \opt(P) - 2n$. Combined with Theorem~\ref{thm:main-sh} this yields $\textsc{RGreedy}(P) \leq 4c \cdot \opt(P) + (8c + 3)n$.  Using Theorem~\ref{thm:perm-all} we can extend this statement to arbitrary access sequences.

\paragraph{Application 3: Recursively decomposable permutations.} 
Let $X$ be a $k$-decomposable permutation. 
Then there is a decomposition tree $\tset$ of $X$ in which each node is a permutation of size at most $k$.
Each block $P \in L(\tset) \cup N(\tset)$ has $\textsc{Greedy}(P) = O(|P| \log |P|)  = O(|P| \log k)$.
By standard charging arguments, the sum $\sum_{P \in N(\tset) \cup L(\tset)} |P| = O(n)$, so the total cost of {\sc RGreedy} is $O(n \log k)$; see the full proof in Appendix~\ref{sec:decomposition}.  

We remark that the logarithmic dependence on $k$ is optimal (Proof in Appendix~\ref{sec:tight-block}), and that \greedy asymptotically matches this bound when $k$ is constant (Theorem~\ref{thm:main-noinitial}).

\newpage 
\section{Discussion and open questions} 

\label{sec:conclusions}

Besides the long-standing open question of resolving dynamic optimality, our work raises several new ones in the context of pattern avoidance. We list some that are most interesting to us.  

Can the bound of Theorem~\ref{thm:pattern-avoidance} be improved? While our input-revealing techniques are unlikely to yield a linear bound (Appendix~\ref{sec:bad-example2}), a slight improvement could come from strengthening Lemma~\ref{lem:forbidden}(ii) for the special kind of light matrices that arise in our proof.

An intriguing question is whether there is a natural characterization of all sequences with linear optimum. How big is the overlap between ``easy inputs'' and inputs with pattern avoidance properties? We have shown that avoidance of a small pattern makes sequences easy.  
The converse does not hold: There is a permutation $X \in S_n$ with $k(X) = \sqrt{n}$ but for which $\textsc{Greedy}(X) = O(n)$; see Appendix~\ref{sec:comparisons}.
Note that our pattern-avoiding properties are incomparable with the earlier parametrizations (e.g.\ dynamic finger); see Appendix~\ref{sec:comparisons}.  
Is there a parameter that subsumes both pattern-avoidance and dynamic finger? 

A question directly related to our work is to close the gap between $\opt = O(n \log k)$ and $n 2^{O(k^2)}$ by \greedy on $k$-decomposable sequences (when $k = \omega(1)$). Matching the optimum (if at all possible) likely requires novel techniques: splay is not even known to be linear on preorder sequences with preprocessing, and with forbidden-submatrix-arguments it seems impossible to obtain bounds beyond\footnote{An $n$-by-$n$ matrix avoiding \emph{all} permutations of size at least $k$ can contain $\Omega(nk)$ ones.} $O(n k)$. 

We proved that if {\sc Greedy} is optimal on simple permutations, then {\sc RGreedy} is optimal on all access sequences. Can some property of simple permutations\footnote{Note that simple permutations can have linear cost. In Appendix~\ref{sec:comparisons} we show such an example.} be algorithmically exploited? %%Can one design, analogously to {\sc RGreedy}, ``robust'' variants of other BST algorithms (e.g.\ splay)? 
Can {\sc RGreedy} be made online? 
Can our application for Cole's showcase sequence be extended in order to prove the dynamic finger property for {\sc Greedy}? 

Finally, making further progress on the traversal conjecture will likely require novel ideas. 
We propose a simpler question that captures the barrier for three famous conjectures: {\em traversal}, {\em deque}, and {\em split} (Appendix~\ref{sec:path-all}). 
Given any initial tree $T$, access a preorder sequence defined by a path $P$ (a BST where each non-leaf node has a single child). 
Prove a linear bound for your favorite online BST!

\newpage 
\bibliographystyle{plain}
\bibliography{ref}

\begin{thebibliography}{10}

\bibitem{bienstock_gyori}
Daniel Bienstock and Ervin Gy{\"{o}}ri.
\newblock An extremal problem on sparse 0-1 matrices.
\newblock {\em {SIAM} J. Discrete Math.}, 4(1):17--27, 1991.

\bibitem{separable}
Prosenjit Bose, Jonathan~F Buss, and Anna Lubiw.
\newblock Pattern matching for permutations.
\newblock {\em Information Processing Letters}, 65(5):277--283, 1998.

\bibitem{brignall}
Robert Brignall.
\newblock A survey of simple permutations.
\newblock {\em CoRR}, abs/0801.0963, 2008.

\bibitem{BrunerL13}
Marie{-}Louise Bruner and Martin Lackner.
\newblock The computational landscape of permutation patterns.
\newblock {\em CoRR}, abs/1301.0340, 2013.

\bibitem{unified}
Mihai B\u{a}doiu, Richard Cole, Erik~D. Demaine, and John Iacono.
\newblock A unified access bound on comparison-based dynamic dictionaries.
\newblock {\em Theoretical Computer Science}, 382(2):86--96, August 2007.
\newblock Special issue of selected papers from the 6th Latin American
  Symposium on Theoretical Informatics, 2004.

\bibitem{our_wads}
P.~Chalermsook, M.~Goswami, L.~Kozma, K.~Mehlhorn, and T.~Saranurak.
\newblock Greedy is an almost optimal deque.
\newblock {\em WADS}, 2015.

\bibitem{our_icalp}
P.~Chalermsook, M.~Goswami, L.~Kozma, K.~Mehlhorn, and T.~Saranurak.
\newblock Self-adjusting binary search trees: What makes them tick?
\newblock {\em ESA}, 2015.

\bibitem{chaudhuri}
R.~Chaudhuri and H.~H\"{o}ft.
\newblock Splaying a search tree in preorder takes linear time.
\newblock {\em SIGACT News}, 24(2):88--93, April 1993.

\bibitem{finger2}
R.~Cole.
\newblock On the dynamic finger conjecture for splay trees. part ii: The proof.
\newblock {\em SIAM Journal on Computing}, 30(1):44--85, 2000.

\bibitem{finger1}
Richard Cole, Bud Mishra, Jeanette Schmidt, and Alan Siegel.
\newblock On the dynamic finger conjecture for splay trees. part i: Splay
  sorting log n-block sequences.
\newblock {\em SIAM J. Comput.}, 30(1):1--43, April 2000.

\bibitem{DHIKP09}
Erik~D. Demaine, Dion Harmon, John Iacono, Daniel~M. Kane, and Mihai Patrascu.
\newblock The geometry of binary search trees.
\newblock In {\em {SODA} 2009}, pages 496--505, 2009.

\bibitem{tango}
Erik~D. Demaine, Dion Harmon, John Iacono, and Mihai Patrascu.
\newblock Dynamic optimality - almost.
\newblock {\em {SIAM} J. Comput.}, 37(1):240--251, 2007.

\bibitem{DemaineILO13}
Erik~D. Demaine, John Iacono, Stefan Langerman, and {\"{O}}zg{\"{u}}r
  {\"{O}}zkan.
\newblock Combining binary search trees.
\newblock In {\em Automata, Languages, and Programming - 40th International
  Colloquium, {ICALP} 2013, Riga, Latvia, July 8-12, 2013, Proceedings, Part
  {I}}, pages 388--399, 2013.

\bibitem{jfox}
Jacob Fox.
\newblock {S}tanley-{W}ilf limits are typically exponential.
\newblock {\em CoRR}, pages --1--1, 2013.

\bibitem{Fox11}
Kyle Fox.
\newblock Upper bounds for maximally greedy binary search trees.
\newblock In {\em {WADS} 2011}, pages 411--422, 2011.

\bibitem{furedi}
Zolt{\'{a}}n F{\"{u}}redi.
\newblock The maximum number of unit distances in a convex \emph{n}-gon.
\newblock {\em J. Comb. Theory, Ser. {A}}, 55(2):316--320, 1990.

\bibitem{chain_splay}
George~F. Georgakopoulos.
\newblock Chain-splay trees, or, how to achieve and prove
  loglogn-competitiveness by splaying.
\newblock {\em Inf. Process. Lett.}, 106(1):37--43, 2008.

\bibitem{Harmon}
Dion Harmon.
\newblock {\em New Bounds on Optimal Binary Search Trees}.
\newblock PhD thesis, Massachusetts Institute of Technology, 2006.

\bibitem{in_pursuit}
John Iacono.
\newblock In pursuit of the dynamic optimality conjecture.
\newblock In {\em Space-Efficient Data Structures, Streams, and Algorithms},
  volume 8066 of {\em Lecture Notes in Computer Science}, pages 236--250.
  Springer Berlin Heidelberg, 2013.

\bibitem{kitaev}
S.~Kitaev.
\newblock {\em Patterns in Permutations and Words}.
\newblock Monographs in Theoretical Computer Science. An EATCS Series.
  Springer, 2011.

\bibitem{knuth68}
Donald~E. Knuth.
\newblock {\em The Art of Computer Programming, Volume {I:} Fundamental
  Algorithms}.
\newblock Addison-Wesley, 1968.

\bibitem{knuth_optimum}
Donald~E. Knuth.
\newblock Optimum binary search trees.
\newblock {\em Acta Informatica}, 1(1):14--25, 1971.

\bibitem{Luc88}
Joan~M. Lucas.
\newblock Canonical forms for competitive binary search tree algorithms.
\newblock {\em Tech. Rep. DCS-TR-250, Rutgers University}, 1988.

\bibitem{split_Luc91}
Joan~M. Lucas.
\newblock On the competitiveness of splay trees: Relations to the union-find
  problem.
\newblock {\em On-line Algorithms, DIMACS Series in Discrete Mathematics and
  Theoretical Computer Science}, 7:95--124, 1991.

\bibitem{marcus_tardos}
Adam Marcus and G\'{a}bor Tardos.
\newblock Excluded permutation matrices and the {S}tanley-{W}ilf conjecture.
\newblock {\em Journal of Combinatorial Theory, Series A}, 107(1):153 -- 160,
  2004.

\bibitem{mehlhorn1975nearly}
Kurt Mehlhorn.
\newblock Nearly optimal binary search trees.
\newblock {\em Acta Informatica}, 5(4):287--295, 1975.

\bibitem{Mun00}
J.Ian Munro.
\newblock On the competitiveness of linear search.
\newblock In Mike~S. Paterson, editor, {\em Algorithms - ESA 2000}, volume 1879
  of {\em Lecture Notes in Computer Science}, pages 338--345. Springer Berlin
  Heidelberg, 2000.

\bibitem{nivasch}
Gabriel Nivasch.
\newblock Improved bounds and new techniques for davenport--schinzel sequences
  and their generalizations.
\newblock {\em J. {ACM}}, 57(3), 2010.

\bibitem{deque_Pet08}
Seth Pettie.
\newblock Splay trees, {D}avenport-{S}chinzel sequences, and the deque
  conjecture.
\newblock In {\em Proceedings of the Nineteenth Annual ACM-SIAM Symposium on
  Discrete Algorithms}, SODA '08, pages 1115--1124, Philadelphia, PA, USA,
  2008. Society for Industrial and Applied Mathematics.

\bibitem{pettie_DS}
Seth Pettie.
\newblock Applications of forbidden 0-1 matrices to search tree and path
  compression-based data structures.
\newblock In {\em Proceedings of the Twenty-First Annual {ACM-SIAM} Symposium
  on Discrete Algorithms, {SODA} 2010, Austin, Texas, USA, January 17-19,
  2010}, pages 1457--1467. {SIAM}, 2010.

\bibitem{pettie_nonlinear}
Seth Pettie.
\newblock Sharp bounds on formation-free sequences.
\newblock In {\em Proceedings of the Twenty-Sixth Annual {ACM-SIAM} Symposium
  on Discrete Algorithms, {SODA} 2015, San Diego, CA, USA, January 4-6, 2015},
  pages 592--604, 2015.

\bibitem{pratt_queues}
Vaughan~R. Pratt.
\newblock Computing permutations with double-ended queues, parallel stacks and
  parallel queues.
\newblock In {\em Proceedings of the Fifth Annual ACM Symposium on Theory of
  Computing}, STOC '73, pages 268--277, New York, NY, USA, 1973. ACM.

\bibitem{schmerl_trotter}
James~H. Schmerl and William~T. Trotter.
\newblock Critically indecomposable partially ordered sets, graphs, tournaments
  and other binary relational structures.
\newblock {\em Discrete Mathematics}, 113(1–3):191 -- 205, 1993.

\bibitem{ST85}
Daniel~Dominic Sleator and Robert~Endre Tarjan.
\newblock Self-adjusting binary search trees.
\newblock {\em J. ACM}, 32(3):652--686, July 1985.

\bibitem{tarjan_sorting}
Robert~Endre Tarjan.
\newblock Sorting using networks of queues and stacks.
\newblock {\em J. ACM}, 19(2):341--346, April 1972.

\bibitem{tarjan_sequential}
Robert~Endre Tarjan.
\newblock Sequential access in splay trees takes linear time.
\newblock {\em Combinatorica}, 5(4):367--378, 1985.

\bibitem{Vatter}
Vincent Vatter.
\newblock Permutation classes.
\newblock {\em CoRR}, abs/1409.5159, 2014.

\bibitem{multisplay}
Chengwen~C. Wang, Jonathan~C. Derryberry, and Daniel~D. Sleator.
\newblock O(log log n)-competitive dynamic binary search trees.
\newblock In {\em Proceedings of the Seventeenth Annual ACM-SIAM Symposium on
  Discrete Algorithm}, SODA '06, pages 374--383, Philadelphia, PA, USA, 2006.
  Society for Industrial and Applied Mathematics.

\bibitem{wilber}
R.~Wilber.
\newblock Lower bounds for accessing binary search trees with rotations.
\newblock {\em SIAM Journal on Computing}, 18(1):56--67, 1989.

\end{thebibliography}

\newpage 
\appendix 
\clearpage

\section{Warmup: A direct proof for preorder sequences}
\label{sec:warmup}

In this section we show that the cost of \greedy on an arbitrary preorder access sequence is linear. In other words, we give a direct proof of Corollary~\ref{cor:trav2}. We assume the preprocessing discussed in \S\,\ref{sec:prelim-short}, or equivalently, in geometric view we assume that \greedy runs with no initial tree\footnote{Curiously, in the case of \greedy and preorder sequences this has the same effect as if the initial tree would be exactly the tree that generated the input sequence.}, i.e.\ initially all stairs are empty.
%%This result can be seen as the geometric analogue of the \emph{traversal conjecture} for splay trees. 
%%The proof of the general preorder access sequences appears in the next section. 
The proof is intended as a warmup, to develop some intuition about the geometric model and to familiarize with the concepts of \emph{wings} and \emph{hidden elements} used in the subsequent proofs. We remark that an alternative proof of the same result (up to the same constants) can be obtained using forbidden submatrix arguments. %Later, wing will be defined slightly differently, but the definitions will capture almost idential ideas. 

\paragraph{Preorder sequences.} The class ${\sf PreOrd}_n$ of preorder sequences of length $n$ can be defined recursively as follows. A permutation sequence $X = (x_1, \dots, x_n)$ belongs to ${\sf PreOrd}_n$ if (i)$|X| = n = 1$ or (ii) there is a number $n' < n$, such that  
\begin{itemize} 
 \item $L(X) = (x_2,\ldots, x_{n'}) \in \sf{PreOrd}_{n'-1}$, 
 \item $R(X) = (x_{n'+1},\ldots, x_n) \in \sf{PreOrd}_{n - n'}$,
% \item $ \max_{x_t \in L} x_t  < x_1 < \min_{x_t \in R} x_t$.
 \item $ \max{\{L(X)\}}  < x_1 < \min{\{R(X)\}}$.
\end{itemize}

We remark again the well-known fact that preorder sequences are exactly the $(2,3,1)$-free permutations, and the easily verifiable fact that preorder sequences are 2-decomposable (see Figure~\ref{fig:preorder-seq} for an illustration).  
\begin{figure}[h]
\centering
\includegraphics[scale=1.4]{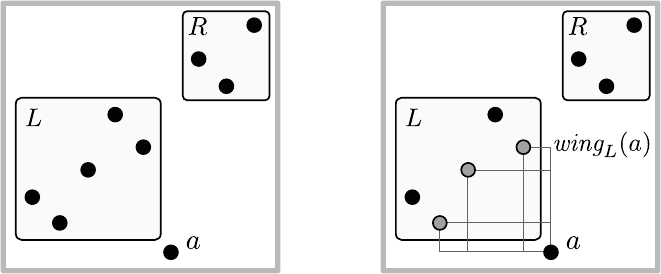}
\caption{Preorder sequence with sets $L$ and $R$, with $\textit{wing}_L(a)$  highlighted. Observe that point $a$ together with $L$ can form a single block.}
\label{fig:preorder-seq}
\end{figure}

We prove the following theorem.

\begin{theorem}~\label{thm:greedy-traversal}
Let $X \in {\sf PreOrd}_n$. Then $\mbox{{\sc Greedy}}(X) \leq 4n$.
\end{theorem}

%\lk{work-in-progress}
%%For element $a$, let $\tset_{a}$ be the subtree rooted at $a$.  

%and let $\wing_{L}(a)$ and $\wing_{R}(a)$ be defined as follows.

%\subsection{The Wings, Hidden Elements, and Forbidden Patterns}
\subsection{Wings and hidden elements}
\label{sec:wing-hidden}

%%The notion of {\em wings} will be crucial for the analysis in Section~\ref{sec:decomposition}.  
  
Let $X \in {\sf PreOrd_n}$. For an arbitrary element $a \in X$, let $t_a$ denote its access time. Let $L_a$ denote the maximal contiguous subsequence of $X$ immediately after $a$, that consists of elements smaller than $a$. Let $r_a \in X$ be the smallest element that is larger than $a$ and precedes $a$ in $X$. By convention, $r_a$ is $\infty$ if there are no elements larger than $a$ that precede $a$ in $X$. Let $R_a$ denote the maximal contiguous subsequence after $L_a$ of elements larger than $a$ and smaller than $r_a$.

%%Let $R_a$ denote the maximal contigous subsequence of $X$ immediately after $L_a$, that consists of elements greater than $a$. 

Let $\tset_a$ be the concatenation of $L_a$ and $R_a$. Let $\textit{lca}(b,c)$ denote the last element $a$, such that $b,c \in \tset_a$.

%\laszlo{should we mention that these correspond to subtrees, or try to do without trees?} 
We define $wing_L(a)$ to contain all elements $b \in L_a$ such that the rectangle with corners $(b, t_b)$, $(a,t_a)$ contains no other access point. Symmetrically, we define $wing_R(a)$ to contain all elements $b \in R_a$ such that the rectangle with corners $(b, t_b)$, $(a,t_a)$ contains no other access point (Figure~\ref{fig:preorder-seq}).

\begin{lemma} 
The sets $\set{ wing_L(a)}_{a \in [n]}$ are disjoint.  
Similarly, the sets $\set{wing_R(a)}_{a \in [n]}$ are disjoint.  
\end{lemma} 
\begin{proof}
Assume for contradiction that there are two elements $a, a' \in [n]$ such that $wing_L(a) \cap wing_L(a') \neq \emptyset$. Assume w.l.o.g.\ that $a<a'$. By definition, $wing_L(a)$ only contains elements in $L_a$, so there must be an element $b \in L_a \cap L_{a'}$, such that $b<a<a'$ and $t_b>t_a$ and $t_b>t_a'$. There are two possibilities: either $t_a > t_a'$, contradicting the fact that the rectangle with corners $(a',t_a')$, $(b,t_b)$ is empty, or $t_a < t_a'$, in which case the subsequence $(a, a', b)$ is order-isomorphic to $(2,3,1)$, contradicting the fact that $X$ is a preorder sequence.
%%This is impossible: In order for $(b,t_b)$ to have empty rectangle with both $(a,t_a)$ and $(a',t_{a'})$, the points $(a,t_a)$ and $(a', t_{a'})$ cannot dominate each other. Assume that $(a,t_a)$ is located bottom-left of $(a',t_{a'})$. Since $b \in L_a \cap L_{a'}$, the point $(b,t_b)$ is located above $t_{a'}$ and left of $a$. This creates a forbidden pattern $\rho$ in $X$, a contradiction.  
%But then both $a$ and $a'$ cannot have $b$ in their subtrees.   
\end{proof} 
The reader might find it instructive to see the corresponding situation in tree-view. Figure~\ref{leftwingpart} shows a proof by picture in the tree corresponding to the preorder sequence $X$.

\begin{figure}[h]
\centering
\includegraphics[width=3.5in]{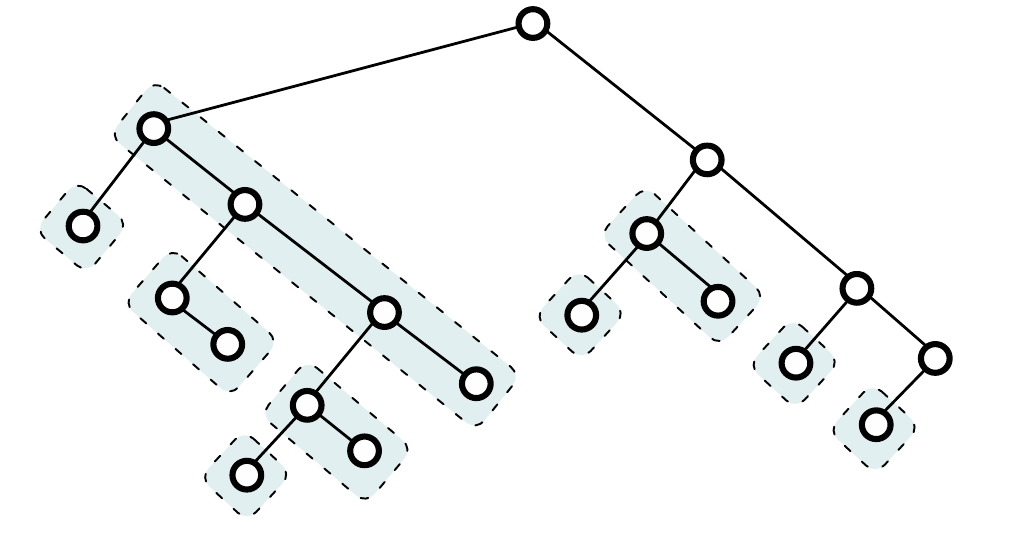}
\caption{BST corresponding to a preorder sequence. Left wings of elements indicated as shaded rectangles.}
\label{leftwingpart}
\end{figure}

\begin{corollary} 
$\sum_a (|wing_L(a)| + |wing_R(a)|) \leq 2n$. 
\end{corollary}

%%The following observations refer to the execution of {\sc Greedy} on the preorder sequence $X$. 

Next we restate and prove Lemma~\ref{prop:hidden} about hidden elements. Recall that $\tau(x,t-1)$ denotes the last time strictly before $t$ when $x$ is touched. 

\begin{definition}[Definition~\ref{def:hidden}]
	For any algorithm $\cA$, an element $x$ is \emph{hidden in} range $[w,y]$ after $t$ if,
	given that there is no access point $a \in [w,y]\times(t,t']$ for any $t'>t$, then
	 $x$ will not be touched by $\cA$ in time $(t,t']$.
\end{definition}

\begin{lemma} [Lemma~\ref{prop:hidden}]
Let $x$ be some element.
\begin{compactenum}[(i)]
\item If there is an element $w<x$ where $\tau(w,t)\ge\tau(x,t)$, then
$x$ is hidden in $(w,n]$ and $(w,x]$ after $t$ for \greedy and \greedyright respectively.
\item If there is an element $y>x$ where $\tau(y,t)\ge\tau(x,t)$, then
$x$ is hidden in $[1,y)$ and $[x,y)$ after $t$ for \greedy and \greedyleft respectively.
\item If there are elements $w,y$ where $w<x<y$, and $\tau(w,t),\tau(y,t)\ge\tau(x,t)$, then
$x$ is hidden in $(w,y)$ after $t$ for \greedy.
\end{compactenum}
\end{lemma}
\begin{proof}
\begin{compactenum}[(i)]
\item Consider any $t'>t$. For the case of \greedy, assume that there is no access point in $(w,n] \times (t,t']$. Suppose $x$ is touched in the time interval $(t,t']$, and let $(p,t_p)$ be the first access point in this time interval that causes the touching of $x$. Then $p \in [1,w]$, and $t_p \in (t,t']$. As $x$ is not touched in the time interval $(t,t_p)$ by the choice of $p$, we have that $\tau(x,t_p-1) = \tau(x,t)$. If $\tau(w,t)\ge\tau(x,t)$, then the rectangle with corners $(p,t_p)$, and $(x,\tau(x,t_p-1))$ contains the point $(w,\tau(w,t))$, and thus it is satisfied before accessing $p$. Therefore, the accessing of $p$ via \greedy does not touch $x$, a contradiction. It follows that $x$ is hidden in $(w,n]$ for \greedy after $t$.

For the case of \greedyright, assume that there is no access point in $(w,x] \times (t,t']$. Suppose $x$ is touched in the time interval $(t,t']$, and let $(p,t_p)$ be the first access point in this time interval that causes the touching of $x$. Then $p \in [1,w]$, and $t_p \in (t,t']$. (The case $p \in (x,n]$ is not possible for \greedyright, since $p<x$ must hold.) The remainder of the argument is the same as for \greedy.

\item The argument is symmetric to (i). 
\item Consider any $t'>t$. Assume that there is no access point in $(w,y) \times (t,t']$. Suppose $x$ is touched in the time interval $(t,t']$, and let $(p,t_p)$ be the first access point in this time interval that causes the touching of $x$. There are two cases. If $p \in [1,w]$, we use the argument of (i). If $p \in [y,n]$, we use the argument of (ii).
\end{compactenum}
\end{proof}

\iffalse
Let $I = [t',t'']$ be a time interval and let $x$ be an element. We say that $x$ is \emph{$[t',t'']$-hidden} if there is a range $R = [v,z] \subseteq [n] $ such that: 
\begin{compactenum}[--]
%%\item $x \in R$ and 
\item no element in $R$ is accessed in $I$. 
\item if there is an access to an element greater than $x$ in $I$, then there is a $y \in (x,z]$ with $\tau(y,t') \ge \tau(x,t')$. 
\item if there is an access to an element smaller than $x$ in $I$, then there is a $w \in [v,x)$ with $\tau(w,t') \ge \tau(x,t')$. 
\end{compactenum}
We also say that \emph{$x$ is $I$-hidden via range $R$} if $R$ satisfies the above.
See Figure~\ref{fig: hidden}.

\begin{lemma} 
\label{lem: hidden lemma}
Let $I$ be a time interval and $x$ be an element that is $I$-hidden. 
Then $x$ is not touched in $I$. \end{lemma}
\begin{proof} Let $I = [t',t'']$ and $R = [v,z]$ and $y$ and $w$ as in the definition of $I$-hidden.

Assume for the sake of contradiction that $x$ is touched during $I$. Let $t \in I$ be the first time, when $x$ is touched during $I$. Then $\tau(t,x) = \tau(t',x)$. We may assume that the access at time $t$ is to an element greater than $z$. By assumption, there is a $y \in (x,z]$ with $\tau(y,t') \ge \tau(x,t')$. Since $\tau(y,t) \ge \tau(y,t')$, 
the rectangle with corners $(x,\tau(x,t)),(z,t)$ contains $(y,\tau(y,t))$ and hence $x$ is not touched at time $t$.
\end{proof}
\fi

\subsection{Bounding the cost of \greedy}  

\begin{lemma}
For any element $c\in L_a \setminus \wing_{L}(a)$, $a$ is not touched when accessing $c$ using \greedy. Similarly for $c \in R_a \setminus wing_R(a)$.  
\end{lemma}
\begin{proof}
There must be an access point $b$, such that $c<b<a$ and $t_a<t_b<t_c$, for otherwise $c$ would be in $\wing_{L}(a)$. Since $\tau(b,t_b) \geq \tau(a,t_b)$, it follows that $a$ is hidden in $(b,n]$ after $t_b$. All accesses in the interval $(t_b, t_c]$ are in $[c,b)$. Hence, $a$ is not touched.
\end{proof}
\begin{corollary}\label{lem_pre1} Let $a$ be any element. 
\greedy touches $a$ when accessing elements in $\tset_a$ at most $|wing_L(a)| + |wing_R(a)|$ times.  
\end{corollary}

\begin{lemma}\label{lem_pre2} Let $a$ be any element. \greedy touches $a$ when accessing elements outside $\tset_a$ at most once. \end{lemma}
\begin{proof}
No element is touched before it is accessed. Since the elements in $\tset_a$ are accessed consecutively, we only need to consider accesses after all elements in $\tset_a$ have been accessed. All such accesses are to elements to the right of $\tset_a$. Let $c$ be the first element accessed after the last access to an element in $\tset_a$ and let $b = \mathit{lca}(a,c)$. Then $a \in L_b$, and $c$ is the first element in $R_b$. 

Since $\tset_c$ contains all elements in $(b,c]$, \greedy does not touch any element in $(b,c]$ before time $t_c$. Let $d$ be the element in $(a,b]$ with largest time $\tau(d,t_c-1)$. If there are several such $d$, let $d$ be the rightmost such element. Since no element in $(b,c]$ is accessed before time $t_c$, $d$ is also the rightmost element in $(a,c]$ with largest time $\tau(d,t_c-1)$. The rectangle with corners $(d,\tau(d,t_c-1))$ and $(c,t_c)$ contains no other point and hence $d$ is touched at time $t_c$. 

Clearly, $\tau(a,t_c) \le t_c = \tau(d,t_c)$. Thus $a$ is hidden in $[1,d)$ after time $t_c$ and hence $a$ is not touched after time $t_c$. 
\end{proof}

Corollary~\ref{lem_pre1} and Lemma~\ref{lem_pre2} together imply that each element $a$ is only touched at most $|wing_L(a)| + |wing_R(a)| +1$ times, so the total cost is at most $\sum_{a \in [n]} (|wing_L(a)| + |wing_R(a)| +2) \leq 4n$.
This concludes the proof of Theorem~\ref{thm:greedy-traversal}.

\section{Access sequences with initial trees} 
\label{sec:forbidden} 

\iffalse
Recall that a sequence $X = (x_1, \dots, x_m)$ \emph{avoids} a permutation $\sigma: [k] \rightarrow [k]$, if there are no indices $1 \leq i_1 < i_2 < \cdots < i_k \leq m$, such that the sequence $x_{i_1}, \dots, x_{i_k}$ is order-isomorphic to $(\sigma(1), \dots, \sigma(k))$. 

In this section we prove that the cost of {\sc Greedy} on an arbitrary access sequence (not necessarily a permutation) that avoids a permutation of size $k$ is quasilinear. The result holds for arbitrary initial tree, therefore, it also holds for the cost of {\sc Greedy} restricted to an arbitrary contiguous subsequence of the input. Since preorder sequences avoid 231, as a special case we obtain a quasilinear bound for the {\sc Greedy} analogue of the traversal conjecture for splay trees. Furthermore, we show that the maximum independent rectangle lower bound (or equivalently the cost of {\sc SignedGreedy}) is linear on any such access sequence. We prove the following theorem.
\fi

\subsection{Initial trees.}
\label{sec:initial-tree}
In the BST model, the \emph{initial tree} is the state of the BST before the first element is accessed. In the tree-view of BST execution, we know the 
structure of the tree exactly in every time step, and hence the structure of
the initial tree as well. This is less obvious, however, in the geometric
view. In fact, according to the reduction of DHIKP, the
initial tree is always a \emph{split tree }\cite{DHIKP09}, hidden from us
in geometric view. We show how to ``force'' an arbitrary initial tree into the geometric view of BST execution.

First, we observe a simple fact. Given a binary search tree $T$, we
denote the search path of an element $a$ in $T$ as ${\sf path}(a,T)$. It is easy to show that two BSTs $T_{1}$ and $T_{2}$,
containing $[n]$, have the same structure iff ${\sf path}(a,T_{1})={\sf path}(a,T_{2})$
as a set, for all $a\in[n]$.

The stair of an element is, in some sense, the geometric view of the
search path. Indeed, ${\sf path}(a,T)=\{b\mid b$ is touched if $a$ is
accessed in $T\}$. Similarly, ${\sf stair}_t(a)=\{b\mid b$
is touched if $a$ is accessed at time $t\}$. Thus, to model an initial
tree $T$ into the geometric view, %by Proposition \ref{same tree=same paths},
it is enough to enforce that ${\sf stair}_1(a)={\sf path}(a,T)$.

We can achieve this in a straightforward way, by putting stacks of points in each column below the first row, corresponding to the structure of the initial tree $T$. More precisely, in column $x$ we put a stack of height $d-d(x)+1$, where $d(x)$ is the depth of element $x$ in $T$, and $d$ is the maximum depth of an element (Figure \ref{fig:geo init tree}). One can easily show that ${\sf stair_1(a)}={\sf path(a,T)}$ for all $a\in[n]$.

We remark that the points representing the initial tree are not counted in the cost.

\subsection{$k$-increasing sequences}
\label{sub:k inc seq}
We prove Lemma~\ref{lem:input-revealing} (iii) and (iv).   

\begin{lemma}[Lemma~\ref{lem:input-revealing}(iii)]
$\inc_{k+1}$ ($\dec_{k+1})$ is a $k$-increasing ($k$-decreasing)
		gadget for \greedy where $\inc_{k}=(k,\dots,1)$ and \textup{$\dec_{k}=(1,2,\dots,k)$}. 
\end{lemma}

\begin{proof}

We only prove that $\inc_{k+1}$ is a $k$-increasing gadget. Let $q_{0},\dots,q_{k}$ be
	touch points in $\Greedy_{T}(X)$ such that $q_{i}.y<q_{i+1}.y$
	for all $i$ and form $\inc_{k+1}$ with bounding box $B$. 
	Note that, for all $1\le i\le k$, we can assume that there is no touch point $q'_i$ where $q'_i.x = q_i.x$ and $q_{i-1}.y<q'_i.y<q_i.y$, otherwise we can choose $q'_i$ instead of $q_i$.
	So, for any $t \in [q_{i-1}.y,q_i.y)$ we have that $\tau(q_{i-1}.x,t)\ge\tau(q_{i}.x,t)$ 
	for all $1\le i\le k$. Suppose that there is no access point in $B$.
	
	By \Cref{prop:hidden}~(ii), $q_{1}.x$ is hidden in $[1,q_{0}.x)$ after $q_0.y$. 
	So there must be an access point $p_1 \in[1,B.\xmin)\times(q_{0}.y,q_{1}.y]$, otherwise $q_1$ cannot be touched.
	We choose such $p_{1}$ where $p_{1}.x$ is maximized. If $p_{1}.y<q_{1}.y$,
	then $\tau(p_{1}.x,p_{1}.y),\tau(q_{0}.x,p_{1}.y)\ge\tau(q_{1}.x,p_{1}.y)$.
	So by \Cref{prop:hidden}~(iii), $q_{1}.x$ is hidden in $(p_{1}.x,q_{0}.x)$ after $p_1.y$ and hence $q_{1}$ cannot be touched, which is a contradiction.
	So we have $p_{1}.y=q_{1}.y$ and $p_{1}.x<B.\xmin$.
	
	Next, we prove by induction from $i=2$ to $k$ that, given the
	access point $p_{i-1}$ where $p_{i-1}.y=q_{i-1}.y$ and $p_{i-1}.x<B.\xmin$,
	there must be an access point $p_{i}$ where $p_{i}.y=q_{i}.y$ and
	$p_{i-1}.x<p_{i}.x<B.\xmin$. Again, by \Cref{prop:hidden}~(iii), there
	must be an access point $p_{i}\in(p_{i-1}.x,q_{i-1}.x)\times(q_{i-1}.y,q_{i}.y]$
	otherwise $q_{i}$ cannot be touched. We choose the point $p_{i}$
	where $p_{i}.x$ is maximized. If $p_{i}.y<q_{i}.y$, then $\tau(p_{i}.x,p_{i}.y),\tau(q_{i-1}.x,p_{i}.y)\ge\tau(q_{i}.x,p_{i}.y)$.
	So by \Cref{prop:hidden}~(iii), $q_{i}.x$ is hidden in $(p_{i}.x,q_{i-1}.x)$ after $p_i.y$ and hence $q_{i}$ cannot be touched.
	So we have $p_{i}.y=q_{i}.y$ and $p_{i-1}.x<p_{i}.x<B.\xmin$. Therefore,
	we get $p_{1}.x<\dots<p_{k}.x<B.\xmin$ which concludes the proof.
\end{proof}

\begin{lemma}[Lemma~\ref{lem:input-revealing}(iv)]
 $\dec_{k+1}$ ($\inc_{k+1})$ is a capture gadget for \greedy that
		runs on $k$-increasing $(k$-decreasing$)$ sequence.
\end{lemma}
\begin{proof}
We only prove that $\inc_{k+1}$ is a capture gadget for \greedy
	running on $k$-decreasing sequence $X$. Since $\inc_{k+1}$ is an
	increasing gadget, if $\inc_{k+1}$ appears in $\Greedy_{T}(X)$ with
	bounding box $B$, then either there is an access point in $B$ or
	there are access points forming $(1,2,\dots,k)$. The latter case
	cannot happen because $X$ avoids $(1,2,\dots,k)$.
	So $\inc_{k+1}$ is a capture gadget.
\end{proof}	

\section{Access sequences without initial trees: $k$-decomposable permutations}
\label{sec:block_linear}

\subsection{Pattern avoidance and recursive decomposability} 
\label{sec:connection-ext}
We restate and prove Lemma~\ref{lem:connection} that connects  the $k$-decomposability of a permutation with pattern avoidance properties.

\begin{lemma}[Extended form of Lemma~\ref{lem:connection}]
\label{lem:connection-ext}
Let $P$ be a permutation, and let $k$ be an integer. The following statements are equivalent:
\begin{enumerate}[label={{(}\roman*{)}}]
\item $P$ is $k$-decomposable,
\item $P$ avoids all simple permutations of size at least $k+1$,
\item $P$ avoids all simple permutations of size $k+1$ and $k+2$.
\end{enumerate}
\end{lemma}

\begin{proof}
$(ii) \implies (iii)$ is obvious.

$(iii) \implies (ii)$ follows from the result of Schmerl and Trotter~\cite{schmerl_trotter} that every simple permutation of length $n$ contains a simple permutation of length $n-1$ or $n-2$, and the simple observation that if $P$ avoids $Q$, then $P$ also avoids all permutations containing $Q$.

$(ii) \implies (i)$ follows from the observation that a permutation $P$ contains all permutations in its block decomposition tree $\tset_P$. Further, it is known~\cite{brignall} that every permutation $P$ has a block decomposition tree $\tset_P$ in which all nodes are simple permutations. If $P$ contains no simple permutation of size $k+1$ or more, it must have a block decomposition tree in which all nodes are simple permutations of size at most $k$, it is therefore, $k$-decomposable.

$(i) \implies (ii)$: we show the contrapositive $\neg(ii) \implies \neg(i)$. 
Indeed, if $P$ contains a simple permutation $Q$ of size at least $k+1$, then any $k$-decomposition of $P$ would induce a $k$-decomposition of $Q$, contradicting the fact that $Q$ is simple.
\end{proof}

\subsection{$k$-decomposable permutations}
\label{sub:k dec seq}
We prove Lemma~\ref{lem:input-revealing} (v) and (vi). We refer to \S\,\ref{sec:input-rev} for the necessary definitions. We show an example of the gadget $\alt_k$: $$\alt_{7} = \left( \begin{smallmatrix}
  &&&\bullet &&&&\\
  &&&&&\bullet &&\\
  &&\bullet &&&&&\\
  &&&&&&\bullet &\\
  &\bullet &&&&&&\\
  &&&&&&&\bullet \\
  &&&&\bullet &&&\\
 \end{smallmatrix}\right).$$

\begin{lemma}[Lemma~\ref{lem:input-revealing}(v)]
		$\alt_{k+1}$ is a $k$-alternating gadget for \greedy where $\alt_{k}=(\left\lfloor (k+1)/2\right\rfloor,k,1,k-1,2,\dots)$. 
\end{lemma}

\begin{proof}

Let $q_{0},\dots,q_{k}$ be
	touch points in $\Greedy(X)$ such that $q_{i}.y<q_{i+1}.y$
	for all $i$ and form $\alt_{k+1}$ with bounding box $B$. 
Suppose that there is no access point in $B$.

	We prove for $i=1,\dots,k$, that there exists an access point $p_i \in (B.\xmax,n]\times(q_{i-1}.y,q_{i}.y]$ for odd $i$, and $p_i \in [1,B.\xmin)\times(q_{i-1}.y,q_{i}.y]$ for even $i$.

For odd $i$, by \Cref{prop:hidden}~(i), $q_{i}.x$ is hidden in $(q_{i-1}.x,n]$ after $q_{i-1}.y$. So there must be an access point $p_i \in(B.\xmax,n]\times(q_{i-1}.y,q_{i}.y]$, otherwise $q_i$ cannot be touched.  

For even $i$, by \Cref{prop:hidden}~(ii), $q_{i}.x$ is hidden in $[1,q_{i-1}.x)$ after $q_{i-1}.y$. 
	So there must be an access point $p_i \in[1,B.\xmin)\times(q_{i-1}.y,q_{i}.y]$, otherwise $q_i$ cannot be touched.

Observe that $p_{1},\dots,p_{k}$ satisfy $B.\ymin \le p_{1}.y<\dots<p_{k}.y\le B.\ymax$
		and $B.\xmax<p_{i}.x$ for all odd $i$ and $p_{i}.x<B.\xmin$ for all
		even $i$, fulfilling the condition of Definition~\ref{def:gadget}. Hence, $\alt_{k+1}$ is a $k$-alternating gadget.

\end{proof}

\begin{lemma}[Lemma~\ref{lem:input-revealing}(vi)]
	 $\alt_{k+4}$ is a capture gadget for \greedy that runs on $k$-decomposable
		sequence without initial tree.
 
\end{lemma}

\begin{proof}

Let $X$ be a $k$-decomposable permutation.
Since $\alt_{k+4}$ is a $(k+3)$-alternating gadget, if $\alt_{k+4}$ appears in $\Greedy(X)$ with bounding box $B$, then either there is an access point in $B$ or there are access points $p_1, p_2, \dots, p_{k+3}$, satisfying the conditions $B.\ymin \leq p_1.y < \cdots < p_{k+3}.y \leq B.\ymax$ and $B.\xmax < p_i.x$ for odd $i$ and $p_i.x < B.\xmin$ for even $i$. Suppose there is no access point in $B$.

Let $q_0, q_1, \dots, q_{k+3}$ denote the points that form $\alt_{k+4}$ (where $q_0.y < q_1.y < \cdots < q_{k+3}.y$), and let us denote $\pset = \{p_1, \dots, p_{k+3}\}$. Let $\tset$ be a block decomposition tree of $X$ of arity at most $k$. We look at each block $P \in \tset$ as a minimal rectangular region in the plane that contains all access points in the block. 

Let $P^*$ be the \emph{smallest} block in $\tset$ that contains two distinct access points $p_i, p_j \in \pset$ such that $i$ is odd, and $j$ is even. (Observe that $p_i$ is to the right of $B$, and $p_j$ is to the left of $B$.) 

Observe first that the bounding box of $P^*$ must contain or intersect both vertical sides of $B$, otherwise it could not contain points on both sides of $B$. 
Furthermore, the bottom horizontal side of $B$ must be contained entirely in the bounding box of $P^*$: If that were not the case, then there would be no accesses below $B$, and $q_0$ (the lowest point of $\alt_{k+4}$) would not be touched by \greedy (since there is no initial tree).
The following is the key observation of the proof.

\begin{lemma}
\label{lem:d5}
Let $k'$ be the largest integer such that $P^*$ contains $q_0, \dots, q_{k'}$. Then, $k' < k+1$. In words, $P^*$ contains at most the $k+1$ lowest points of $\alt_{k+4}$.
\end{lemma}

\begin{proof}
Let $P^* = \tilde P(P_1,\ldots, P_k)$ be the decomposition of $P^*$ into at most $k$ blocks.
We claim that each block $P_j$ can contain at most one point from $\pset$. First, by the minimality of $P^*$, each block $P_j$ never contains two access points from $\pset$ from different sides of $B$; and $P_j$ cannot contain two points from $\pset$ from the same side of $B$ either, because assuming otherwise that $P_j$ contains two points from $\pset$ from the left of $B$, it would follow that there is another access point $(a,t) \in \pset$ on the right of $B$ outside of $P_j$ such that $(P_j).\ymin < t< (P_j).\ymax$, contradicting that $P_j$ is a block.  

Observe further that, except for the block containing $p_1$, each block $P_j$ cannot overlap with $B$. %% (we remark that the block $P_j$ that contains $p_1$ can overlap with the corner of $B$.)   
Since we have at most $k$ such blocks in the decomposition, and because $q_i.y \ge p_i.y$ for all $i$, it follows that the top boundary of $P^*$ must be below $q_{k+1}$.  
 
\end{proof}

Since the bounding box of $P^*$ contains (at best) the first $k+1$ points of $\alt_{k+4}$, the three remaining points of $\alt_{k+4}$ need to be touched after the last access in $P^*$. In the following we show that this is impossible.

Let $(L, t_L)$ be the topmost access point inside $P^*$ to the left of $B$, such that there is a touch point inside $B$ at time $t_L$. Let $(R, t_R)$ be the topmost access point inside $P^*$ to the right of $B$, such that there is a touch point inside $B$ at time $t_R$. Let $L'$ be the rightmost element inside $B$ touched at time $t_L$, and let $R'$ be the leftmost element inside $B$ touched at time of $t_R$.

\begin{lemma}
\label{lem:onlybetween}
After time $t_L$, within the interval $[B.\xmin, L']$, only $L'$ can be touched. Similarly, after time $t_R$, within the interval $[R', B.\xmax]$, only $R'$ can be touched. 
\end{lemma}

\begin{proof}
Let $L''$ be the rightmost touched element in $[L,B.\xmin)$ at time $t_L$, and let $R''$ be the leftmost touched element in $(B.\xmax, R]$ at time $t_R$. 

%, and no access to an element in $(B.\xmax, R'')$,
%$Let $t^* = (P^*).\ymax$. 
We have that $\tau(L'',t_L), \tau(L',t_L) \geq \tau(x,t_L)$, for any $x \in [B.\xmin, L')$. Thus, any $x \in [B.\xmin, L')$ is hidden in $(L'',L')$ after $t_L$.

Above $P^*$ there can be no access to an element in $[L'',L']$, since that would contradict the fact that $P^*$ is a block. Within $P^*$ after $t_L$ there can be no access to an element in $(L'',B.\xmin)$, as the first such access would necessarily cause the touching of a point inside $B$, a contradiction to the choice of $(L,t_L)$. An access within $B$ is ruled out by our initial assumption.

Thus, there is no access in $(L'',L') \times (t_L,n]$, and from Lemma~\ref{prop:hidden}(iii) it follows that any $x \in [B.\xmin, L')$ can not be touched after time $t_L$.

We argue similarly on the right side of $B$.
\end{proof}

As a corollary of Lemma~\ref{lem:onlybetween}, note that above $P^*$ only elements within $[L',R']$ can be touched within $B$. %%We will argue that at most two different elements of this interval are touched outside of $P^*$. 
If $L' = R'$, then only one element can be touched, and we are done. Therefore, we assume that $L'$ is strictly to the left of $R'$.

We assume w.l.o.g.\ that $t_R > t_L$ (the other case is symmetric). Let $(Q,t_Q)$ be the first access point left of $B$ with $t_Q > t_R$ such that there is a touch point inside $B$ at time $t_Q$. Observe that $(Q,t_Q)$ is outside of $P^*$ by our choice of $(L,t_L)$. If no such $Q$ exists, we are done, since we contradict the fact that $\alt_{k+4}$ is a $(k+3)$-alternating gadget. 

Let $(Z, t_Z)$ denote the last touch point with the property that $Z \in [(P^*).xmin, L')$, and $t_Z \in [t_L, t_Q)$. If there are more such points in the same row, pick $Z$ to be the leftmost. Note that $(Z, t_Z)$ might coincide with $(L, t_L)$. We have $t_Z \le t_R$ because otherwise $(Q,t_Q)$ cannot exist. Note than $t_Z > t_R$ would imply that elements in $[L',R']$ are hidden in $(Z,R)$ after $t_R$. 

Next, we make some easy observations.

\begin{lemma}
\label{lem:emptiness}
The following ranges are empty:
\begin{enumerate}[label={{(}\roman*{)}}]
\item $[B.\xmin, B.\xmax] \times [t_R+1, t_Q-1]$,
\item $[1, Z-1] \times [t_Z+1, t_R]$,
\item $[Q, Z-1] \times [t_Z+1, t_Q-1]$,
\item $[Q, R'-1] \times [t_{R}+1, t_Q-1]$.
\end{enumerate}
\end{lemma} 

\begin{proof}

\begin{enumerate}[label={{(}\roman*{)}}]
\item It is clear that there can be no access point inside $B$ by assumption. Suppose there is a touch point in $[B.\xmin, B.\xmax] \times [t_R+1, t_Q-1]$, and let $(Q',t_{Q'})$ be the first (lowest) such point. Denote the access point at time $t_{Q'}$ as $(Q'', t_{Q'})$. Clearly $t_{Q'}>(P^*).\ymax$ must hold, otherwise the choice of $(L,t_L)$ or $(R,t_R)$ would be contradicted. Also $Q''>(P^*).\xmax$ must hold, otherwise the choice of $(Q,t_Q)$ would be contradicted.

But this is impossible, since $Q' \in [L',R']$, and $\tau(Q',t_{Q'}-1) \leq t_R$, and thus the rectangle with corners $(Q'', t_{Q'})$ and $(Q', \tau(Q',t_{Q'}-1))$ contains the touch point $(R,t_R)$, contradicting the claim that \greedy touches $Q'$ at time $t_{Q'}$. 
\item 
All elements in $[1, Z-1]$ are hidden in $[1,Z-1]$ after $t_Z$. There can be no access points on the left of $P^*$, due to the structure of the block decomposition, and there is no access point in $[(P^*).\xmin, Z-1] \times [t_Z+1, t_Q-1]$ due to the choice of $Z$.  
%No touched points because of hiding arguments and the choice of $Z$.
%Moreover, during the time $[t_Z+1, t_R]$, all accesses are on the right of $Z$ (again from the choice of $(Z,t_Z)$), and 
Hence there cannot be a touch points in $[1,Z-1] \times [t_Z+1, t_R]$.  
\item First there is no access point in $[Q,Z-1] \times [t_Z+1,(P^*).\ymax]$ due to the choice of $(Z,t_Z)$ and the structure of block decomposition. 
Also, there is no access point in $[Q,Z-1] \times ((P^*).\ymax, t_Q-1]$: Assume there were such an access point $(Q',t_{Q'})$. Then it must be the case that $Q < Q' < (P^*).\xmin$ and $(P^*).\ymax < t_{Q'} < t_Q$. 
Any rectangle formed by $(Q,t_Q)$  and a point in $P^* \cap B$ would have contained $(Q',t_{Q'})$, a contradiction to the fact that \greedy touches a point inside $B$ at time $t_Q$. 

Next, we argue that there is no non-access touch point $(a,t) \in [Q,Z-1] \times [t_Z+1,t_Q-1]$. 
There are three cases.  
\begin{compactenum}[--]
\item $(a,t) \in [(P^*).\xmin, Z-1] \times [t_Z+1,t_Q-1]$ contradicts the choice of $(Z,t_Z)$.

\item $(a,t) \in [Q,(P^*).\xmin) \times [t_Z+1, (P^*).\ymax]$ contradicts the fact that all elements in $[1,Z)$ are hidden in $[1,Z)$ after $t_Z$, and there is no access in $[1,Z)$ in the time interval $[t_Z+1,(P^*).\ymax]$, since $P^*$ is a block. 
%% if $a>(P^*).\xmin$, we contradict the choice of $(Z,t_Z)$. 
%%If $a < (P^*).\xmin$, we know that $a$ is $[t_Z+1, \mtop(P^*)]$-hidden via $[0,Z-1]$.  
\item $(a,t) \in [Q,(P^*).\xmin) \times ((P^*).\ymax,t_Q-1]$, contradicts the claim that at time $t_Q$ \greedy touches a point inside $B$, since any rectangle formed by $(Q,t_Q)$  and a point in $P^* \cap B$ would have contained $(a,t)$.
\end{compactenum} 
 
%Hiding arguments and the choice of $Q$.
\item
Given the previous claims, it remains only to prove that there is no touch point $(a,t) \in [Z, B.\xmin) \times t \in [t_R+1, t_Q-1]$.
There cannot be such a touch point for $t \leq (P^*).\ymax$ due to the choice of $(Z,t_Z)$. For $t > (P^*).\ymax$, there cannot be souch an access point due to the structure of block decomposition. Remains the case when $(a,t)$ is a non-access touch point in $[Z,B.\xmin) \times ((P^*).\ymax, t_Q-1$. This is also impossible, as any rectangle formed by $(Q,t_Q)$ and a point in $P^* \cap B$ would have contained $(a,t)$, contradicting the choice of $(Q,t_Q)$. 

%Hiding arguments and the choice of $R$, $R'$, $Q$.
\end{enumerate} 
\end{proof}

\begin{lemma}
\label{lem:d8}
At time $t_Q$ %if we touch an element in $[L',R']$, then 
we touch $R'$. Let $L^*$ be the leftmost element touched in $[L',R']$ at time $t_Q$ (it is possible that $L^* = R'$). After time $t_Q$, only $L^*$ and $R'$ can be touched within $[L',R']$. 
%%%, because $[L',L'')$ is blocked by $X$. If $L''$ does not exist, it means that $X$ blocks $[L',R')$, so only $R'$ can be touched.
\end{lemma}
\begin{proof}
The first part follows from the emptiness of rectangles in Lemma~\ref{lem:emptiness}.
That is, the lemma implies that the rectangle formed by $(Q,t_Q)$ and $(R', t_R)$ is empty, so \greedy touches $R'$ at time $t_Q$.  

%The second part follows from hiding arguments.
The elements $a \in (L^*, R')$ cannot be touched because they are hidden in $(L^*,R')$ after $t_Q$, and there is no access point in this range after $t_Q$, due to the structure of the block decomposition.
Suppose that some element $a \in [L', L^*)$ is touched (for the first time) at some time $t > t_Q$ by accessing element $x$. 
So the rectangle formed by $(x,t)$ and $(a,\tau(a,t-1))$ is empty, and we know that $\tau(a,t-1) < t_R$. Notice that $x<(P^*).\xmin$ for otherwise $(R,t_R)$ would have been in the rectangle formed by $(x,t)$ and $(a,\tau(a,t-1))$, a contradiction.
Furthermore, $\tau(a,t-1) > t_Z$ for otherwise, $(Z,t_Z)$ would have been in the rectangle formed by $(x,t)$ and $(a,\tau(a,t-1))$, a contradiction. 
Now, since $\tau(a,t-1) > t_Z$, Lemma~\ref{lem:emptiness} implies that the rectangle formed by $(Q,t_Q)$ and $(a,\tau(a,t-1))$ must be empty, therefore $a$ is touched at time $t_Q$; this contradicts the choice of $L^*$. 

\end{proof}

Lemma~\ref{lem:d5}, \ref{lem:onlybetween}, and \ref{lem:d8} together imply that in total at best the $k+3$ lowest points of $\alt_{k+4}$ can be touched (out of the $k+4$ needed). This means that our assumption that $B$ is free of access points was false. Therefore $\alt_{k+4}$ is a capture gadget, finishing the proof of Lemma~\ref{lem:input-revealing}(vi). See Figure~\ref{proof_slide} for an illustration.

\begin{figure}
\begin{center}  
\includegraphics[scale=0.7]{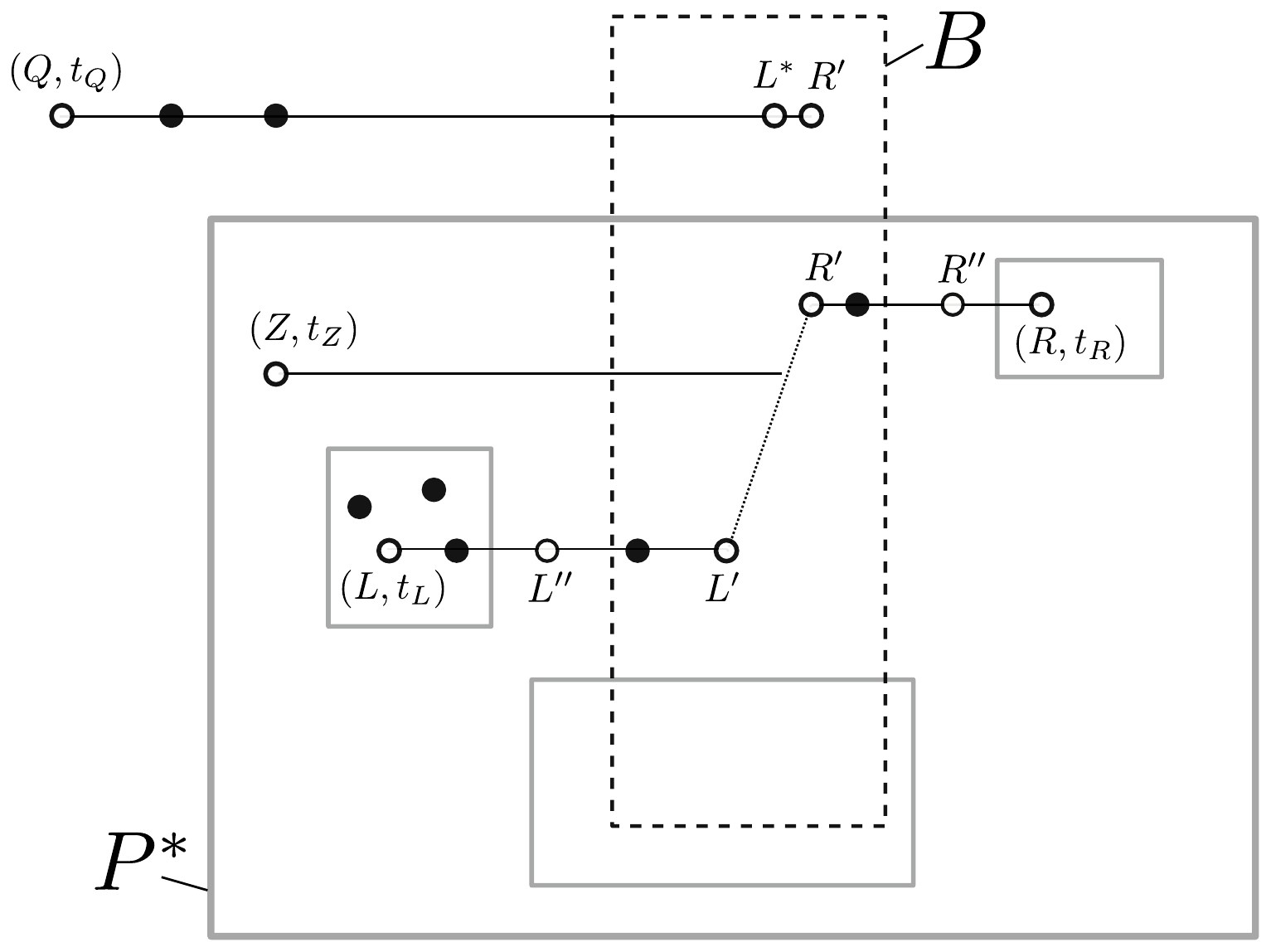}
\end{center}
\caption{Illustration of the proof of Lemma~\ref{lem:input-revealing}(vi).} 
\label{proof_slide}
\end{figure}
\end{proof}

\iffalse
\subsection{Alternative direct proof for preorder sequences}
\label{sec:alt-direct}

We give an alternative proof of Theorem~\ref{thm:greedy-traversal}.

\begin{lemma}(F\"{u}redi, Hajnal~\cite[Cor.\ 2.7]{furedi_hajnal})~~ Let $M$ be an $n \times m$ matrix. If $M$ is $(231)$-free, then $w(M) \leq 2(m+n)$.
\end{Lemma}

Recall that a preorder sequence is $231$-free, and observe that $(231) = \alt_3$. Let $X \in {\sf PreOrd}_n$. We use Lemma~\ref{lem:input-revealing} (v).
\fi

\section{Decomposition theorem and applications}
\label{sec:decomposition}

\subsection{Preliminaries}
\label{subsec:decomp}
Given a permutation $\sigma: [n] \rightarrow [n]$, we call a set $[a,b] \times [c,d]$, with the property that $\{\sigma(a), \sigma(a+1), \dots, \sigma(b)\} = [c,d]$, a \emph{block} of $\sigma$ %of size $b-a+1$
\emph{induced} by the interval $[a,b]$. 
Let $P$ be a block induced by interval $I$. Consider a partitioning of $I$ into $k$ intervals, denoted (left-to-right) as $I_1, I_2, \dots, I_k$, such that each interval induces a block, denoted respectively as $P_1, P_2, \dots, P_k$. Consider the unique permutation $\tilde{P}$ of size $[k]$ that is order-isomorphic to a sequence of arbitrary representative elements from the blocks $P_1, P_2, \dots, P_k$. We denote $P = \tilde{P}[P_1, P_2, \dots, P_k]$ a \emph{deflation} of $P$, and we call $\tilde{P}$ the skeleton of $P$. Visually, $\tilde{P}$ is obtained from $P$ by contracting each block $P_i$ to a single point (Figure~\ref{fig_intro}). Deflation is analogously defined on permutations. Every permutation $\sigma$ has the trivial deflations $\sigma = \sigma[1, \dots, 1]$ and $\sigma=1[\sigma]$. Permutations that have only the trivial deflations are called \emph{simple permutations}.

%%We call the unique permutation of $[j-i+1]$ that is order-isomorphic to $(\sigma(i), \dots, \sigma(j))$, a \emph{block} of $\sigma$ \emph{induced} by the interval $[i,j]$. 

A \emph{block decomposition tree} $\tset$ of $\sigma$ is a tree in which each node is a block of $\sigma$. We define $\tset$ recursively as follows: The root of $\tset$ is the block $[n] \times [n]$. A node $P$ of $\tset$ is either a leaf, or it has children $P_1, P_2, \dots$ corresponding to a nontrivial deflation $P = \tilde{P}[P_1, P_2, \dots ]$. To every non-leaf block $P$ we associate its skeleton $\tilde{P}$. For convenience, we assume that the blocks $P_1,\ldots, P_k$ are ordered by their $y$-coordinates from bottom to top. We partition $P$ into $k^2$ rectangular regions whose boundaries are aligned with $P_1, \ldots, P_k$, and index them $R(i,j)$ with $i$ going left to right and $j$ going bottom to top. In particular, the region $R(i,j)$ is aligned with $P_j$ (same $y$-coordinates). 
To every leaf block $P = [a,b] \times [c,d]$ we associate the unique permutation that is order-isomorphic to $(\sigma(a), \sigma(a+1), \dots, \sigma(b))$. When there is no chance of confusion, we refer interchangeably to a block and its associated permutation. We denote the leaves and non-leaves of the tree $\tset$ by $L(\tset)$ and $N(\tset)$ respectively.

A permutation is \emph{$k$-decomposable} if it has a block decomposition tree $\tset$ in which every node has an associated permutation of size at most $k$. As an example, we observe that preorder sequences of binary search trees are $2$-decomposable.

%%%In the following we mostly look at permutations in the geometric view. Therefore, if a node $P \in B(\tset)$ represents a block of $\sigma$ induced by the interval $[i,j]$, then by slight abuse of notation we will also refer to the set $[i,j] \times \{\sigma(i), \dots, \sigma(j)\}$ as \emph{the block} $P$. \lk{todo: simplify}

%In this case, we say that $P = \tilde{P}[P_1,P_2,P_3,P_4]$ is a \emph{block-decomposition} of $P$, and $P$ is called an {\em inflation} of $\tilde P$.

%If there are $k$ blocks, then the skeleton $\tilde P$ is isomorphic to a permutation over $[k]$. Given a block-decomposition of a permutation, we further decompose each non-simple block, resulting in a tree.

%It is known that any non-simple permutation can be uniquely rewritten as a simple permutation over blocks of permutations (to be formalized later). 
%We will use this fact to prove our theorem, stated formally below.  

\subsection{A robust {\sc Greedy}}

We augment \greedy with some additional steps and we obtain an offline algorithm (called \agreedy) with several desirable properties. We remark that in this section we always assume the preprocessing discussed in \S\,\ref{sec:prelim-short} or equivalently, in geometric view we assume that \agreedy is executed without initial tree (since \agreedy is offline, this is a non-issue).

Before describing the algorithm, we give some simple definitions.

We denote the top left, respectively top right corner point of a rectangular region $R$ as $\textit{topleft}(R)$, respectively $\textit{topright}(R)$. We define $\textit{topwing}(R)$ to be a subset of the points touched in $R$, that are the last touch points of their respective columns. A point $q$ is in $\textit{topwing}(R)$ if at least one of the following properties is satisfied:  %, and 2) the topmost points in the left and right columns of $R$. The top wing of $R$, denoted $\textit{topwing}(R)$ is the set of elements $q$ that are touched inside $R$, and for which at least one of the following properties is satisfied: 
\begin{compactenum}[--]
\item The rectangle with corners $q$ and $\textit{topleft}(R)$ is empty,  
\item The rectangle with corners $q$ and $\textit{topright}(R)$ is empty,  
\item $q$ is in the leftmost or rightmost column of $R$.
\end{compactenum}
Note that the set $\textit{topwing}(R)$ depends on the time point where we are in the execution of the algorithm, but for simplicity of notation, we omit the reference to a specific time (it will always be clear from the context).

We can now describe how \agreedy processes an input point $(a,t)$. First, it acts as \greedy would, i.e.\ for every empty rectangle with corners $(a,t)$ and $(b,t')$, it touches the point $(b,t)$ on the current timeline. Let $Y_t$ be the points touched in this way (including the point $(a,t)$). Second, it touches additional points on the timeline. Let $v_1$, $v_2 = \parent(v_1)$, \ldots, $v_\ell = \parent(v_{\ell - 1})$ be the nodes of the decomposition tree such that the blocks associated with these nodes end at time $t$. Note that $v_1$ is a leaf. Let $v_{\ell +1} = \parent(v_\ell)$; $v_{\ell + 1}$ exists except when the last input is processed. For $h \in [2,\ell + 1]$ in increasing order, we consider the deflation $\tilde{P}[P_1,\ldots,P_k]$ associated with $v_h$. Let $(a,t)$ lie in $P_j$; note that $j = k$ if $h \le \ell$ and $j < k$ for $h = \ell + 1$. For every sibling region $R(i,j)$ of $P_j$, \agreedy touches the projection of $\topwing(R(i,j))$ on the timeline $t$.

Let $Y_{h,t}$ be the points touched. Figure~\ref{fig_rgreedy} illustrates the definition of \agreedy.

%\kurt{The figure needs more explanation. The access point in the topmost line completes the black square and also the enclosing grey square. The topwing of the black square consists of the black circle at position (1,6) and the grey circle at position (3,5). Therefore the augmentation of the black square adds the grey square at position (3,6). The top-wing of $R(2,2)$ consists of the grey circle at position (4,4) and hence the grey square at position (4,6) is added. The topwing of the grey square consists of points (1,6) and (6,3); the point (4,4) does not belong to the topwing because (4,6) has already been added. Therefore, the grey square at position (6,6) is added.}

%let $t_{i,j}$ be the topmost $y$-coordinate at which at least one point was touched in $R(i,j)$, and let $p_{left}(i,j), p_{right}(i,j)$ be the leftmost and rightmost points touched in $R(i,j)$ at time $t_{i,j}$.
%We add two points $p'_{left}(i,j), p'_{right}(i,j)$ on the topmost $y$-coordinate of $R(i,j)$ such that $p'_{left}(i,j), p'_{right}(i,j)$ are vertically aligned with $p_{left}(i,j), p_{right}(i,j)$ respectively.   

We need to argue that the augmentation step does not violate the feasibility of the solution.

\begin{lemma}
\label{lem:feasibility}
The output of \agreedy is feasible.  
\end{lemma}  
\begin{proof}

Suppose for contradiction that at time $t$ we have created an unsatisfied rectangle $r$ with corners $(a,t)$ and $(x,y)$, where $y<t$. The point $(a,t)$ is the result of projecting $(a,t')$ onto the horizontal line $t$, where $(a,t') \in \textit{topwing}(R)$, for some region $R$ processed at time $t$. Clearly $t'<t$. Assume that before adding $(a,t)$, the rectangle formed by the points $(a,t')$ and $\textit{topleft}(R)$ is empty. A symmetric argument works if $(a,t')$ formed an empty rectangle with $\textit{topright}(R)$. We have $y > t'$, for otherwise the point $(a,t')$ would satisfy $r$. Let $R = R(i,j)$. Since $(x,y)$ lies in the time window of $R$, we have $(x,y) \in R(i',j)$ for some $i'$.

Suppose first that $i = i'$. 
If $x \leq a$, then the rectangle with corners $(a,t')$ and $\textit{topleft}(R)$ contains $(x,y)$, contradicting the fact that $(a,t')$ was in $\textit{topwing}(R)$ before the augmentation. Therefore $x>a$ holds. But then, if $r$ is not satisfied, then the rectangle with corners $(x,y)$ and $\textit{topleft}(R)$ had to be empty before the augmentation, and therefore $(x,t)$ must have been touched in the augmentation step, making $r$ satisfied. 

The case $i \not= i'$ remains. If $i'<i$, then $r$ contains the rectangle formed by $(x,y)$ and the top right corner of $R(i',j)$. Either this rectangle is nonempty, or $(x,y)$ was in $\textit{topwing}(R(i',j))$ before the augmentation, in which case $(x,t)$ was touched in the augmentation step. In both cases $r$ must be satisfied. The case when $i'>i$ is symmetric, finishing the proof.

\end{proof}

Observe that \agreedy processes the input and emits the output line-by-line, in increasing order of time. It is, however, not an online algorithm, as it has a priori access to a block decomposition tree of the input.

\subsection{Decomposition theorem }

%Now we define the recursive partitioning of the points into {\em blocks}. 
%We will keep track of these blocks by a tree $\tset = (\bset, \pi)$ where $\pi$ is a parent relationship and $\bset$ is a collection of blocks.  
%Let $X$ be the set of initial input points.  
%Our procedure starts from, initially, $\bset = \{X\}$. 
%Then, as long as there is a block $P \in \bset$ that is still not simple, we break $P$ into $P = \tilde P[P_1,\ldots, P_k]$, and we insert $P_1,\ldots, P_k$ into $\bset$. 
%Call these new blocks the children of $P$ and $\tilde P$ the skeleton of $P$.
%Notice that this procedure naturally defines the tree structure, and except for the leafs, all other blocks $P$ in $\bset$ are associated with their skeletons $\tilde P$.  

%\agreedy would add some points to the plane, and we will bound the number of these points. 

%For leaf blocks, the cost will be bounded by our hypothesis that \greedy is $c$-competitve for simple permutations, so we only need to bound the cost of non-leaf blocks. 

Our main technical contribution in this section is the following result: 

\begin{theorem} [Decomposition theorem for permutations] 
\label{lem: main} 
For any block decomposition tree $\tset$ of a permutation $P$, the total cost of \agreedy is bounded as \[ \textsc{RGreedy}(P) \leq 4 \cdot \sum_{\tilde P \in N(\tset) }  \textsc{Greedy}(\tilde P) + \sum_{P \in L(\tset) } \textsc{Greedy}(P)+ 3n.\]
\end{theorem} 

%%%\paragraph{Observation.} \laszlo{now I think this will not work. There are too many places in the proof where we assume properties of Greedy. We probably need to retract this observation, and wherever marked with *} In the description of {\sc RGreedy} we can replace {\sc Greedy} with any other {\sc MinASS} algorithm, and all the proofs in this section go through unchanged, providing results analogous to Theorem~\ref{lem: main}. In this way we can ``robustify'' any {\sc MinASS} algorithm. This observation will become more interesting in light of Application 2\ below.\\ 

We discuss applications of Theorem~\ref{lem: main}.  

\paragraph{Application 1: Cole et al.'s ``showcase'' sequences.}
In Cole et al.~\cite{finger1}, a showcase sequence is defined in order to introduce necessary technical tools for proving the dynamic finger theorem. 
This sequence can be defined (in our context) as a permutation $P = \tilde P[P_1,\ldots, P_{n/\log n}]$ where each $P_j$ is a sequence for which \greedy has linear linear access cost, i.e.\ $|P_j| = \log n$, and $\textsc{Greedy}(P_j) = O(\log n)$. 
Now, by applying Theorem~\ref{lem: main}, we can say that the cost of \agreedy on $P$ is $\textsc{RGreedy}(P) \leq 4 \cdot \textsc{Greedy}(\tilde P) + \sum_{j} \textsc{Greedy}(P_j) + O(n) \leq O(\frac{n}{\log n} \cdot \log (\frac{n}{\log n})) + \frac{n}{\log n}O(\log n) + O(n) = O(n)$. 
This implies that our algorithm has linear access cost on this sequence.

\paragraph{Application 2: Simple permutations as the ``core'' difficulty of BST inputs.}  
We argue that whether \agreedy performs well on all sequences depends only on whether \greedy does well on simple permutation sequences. 
More formally, we prove the following theorem. 

\begin{theorem}
\label{thm:greedy_rgreedy} 
If \greedy is $c$-competitive on simple permutations, then \agreedy is an $O(1)$ approximation algorithm for $\opt(X)$ for all permutations $X$. 
%\kurt{I find the usage of the term $O(c)$-competitive for \agreedy inappropriate. It is an offline-algorithm and hence we should use $O(c)$-approximate instead.}
\end{theorem} 

Using Theorem~\ref{thm:perm-all} we can extend this statement from all permutations to all access sequences.
%%%, and using the observation after Theorem~\ref{lem: main} we can conclude that the task of finding a dynamically optimal BST algorithm has been reduced to the task of finding a BST algorithm that is competitive on simple permutations.\lk{*}\\

To prove Theorem~\ref{thm:greedy_rgreedy} we use of the known easy fact~\cite{brignall} that every permutation $P$ has a block decomposition tree $\tset$ in which all permutations in $L(\tset)$ and all permutations in $N(\tset)$ are simple. Applying Theorem~\ref{lem: main} and the hypothesis that \greedy is $c$-competitive for simple permutations, we get the bound: 
\[\textsc{RGreedy}(P) \leq 4c \cdot \sum_{\tilde P \in  N(\tset)} \opt(\tilde P) + c \cdot \sum_{P \in L(\tset)} \opt(P) +3n.\]

%%starting from the bound $\textsc{Greedy}(P) \leq \sum_{\tilde P: N(\tset) } \textsc{Greedy}(\tilde P) + \sum_{P: N(\tset)} \textsc{Greedy}(P)$,  
%we can apply the hypothesis that \greedy is 

Now we need to decompose $\opt$ into subproblems the same way as done for \agreedy. 
The following lemma finishes the proof (modulo the proof of Theorem~\ref{lem: main}), since it implies that $\textsc{RGreedy}(P) \leq 4c \cdot \opt(P) + (8c + 3)n$.  

\begin{lemma} 
\label{lem:decomp opt lower-bound}
$\opt(P) \geq \sum_{P \in \tset } \opt(P) - 2n$.
\end{lemma}  

\begin{proof} 

We show that for an arbitrary permutation $P = \tilde{P}[P_1,\dots,P_k]$, it holds that $\opt(P) \geq \opt(\tilde{P}) + \sum_{i=1}^k{\opt(P_i)} - k$. The result then follows simply by iterating this operation for every internal node of the block decomposition tree $\tset$ of $P$ and observing that the numbers of children of each internal node add up to at most $2n$. %$\sum_{\tilde P \in N(\tset) }{|\tilde{P}|} \leq n$.\\

Let $Y^*$ be the set of points of the optimal solution for $P$. We partition $Y^*$ into sets $Y^*_i$ contained inside each block $P_i$, and we let $Y^*_{out} = Y^* \setminus \bigcup_{i=1}^k{Y^*_i}$. Each $Y^*_i$ must clearly be feasible for $P_i$, so we have $|Y^*_i| \geq \opt(P_i)$. It remains to show that $|Y^*_{out}| \geq OPT(\tilde{P}) - k$. 
%%% in the same fashion as done in defining $\tset$: For each leaf $P \in L(\tset)$, $Y^*_P$ contains points in $Y^* \cap P$, and for each non-leaf $P \in N(\tset)$ that has the children $P_1, \dots, P_k$,  $Y^*_P$ contains points in $Y^* \cap P \setminus \bigcup_{i=1}^k{P_i}$.

%%Let $Y^*$ be the set of points added by the optimal solution. 
%%We break $Y^*$ into sets $Y^*_P$ in the same fashion as done in defining $\tset$: For each leaf $P \in L(\tset)$, $Y^*_P$ contains points in $Y^* \cap P$, and for each non-leaf $P \in N(\tset)$ that has the children $P_1, \dots, P_k$,  $Y^*_P$ contains points in $Y^* \cap P \setminus \bigcup_{i=1}^k{P_i}$.
%We have that $\sum_{P \in B(\tset)} |Y^*_P| = |Y^*| = \opt$.

%For a leaf block $P$, it is clear that 

Again, consider the partitioning of $P$ into $k^2$ rectangular regions $R(i,j)$, aligned with the blocks $P_1, \dots, P_k$. Consider the contraction operation that replaces each region $R(i,j)$ with a point $(i,j)$ if and only if $Y^* \cap R(i,j)$ is nonempty. We obtain a point set $Y' \in [k] \times [k]$, and clearly $|Y'| \leq |Y^*_{out}| + |\tilde{P}|$. We claim that the obtained point set $Y'$ is a satisfied superset of $\tilde{P}$, and therefore $|Y^*_{out}| \geq \opt(\tilde{P}) - k$. 

Indeed, suppose for contradiction that there is an unsatisfied rectangle with corner points $p,q \in Y'$. Suppose $p$ is to the upper-left of $q$ (the other cases are symmetric). Let $R_p$ and $R_q$ be the regions before the contraction that were mapped to $p$, respectively $q$. Since $Y^* \cap R_p$ is nonempty, we can consider $p'$ to be the lowest point in $Y^* \cap R_p$ and if there are several such points, we take $p'$ to be the rightmost. Similarly, since $Y^* \cap R_q$ is nonempty, we can take $q'$ to be the topmost point in $Y^* \cap R_q$ and if there are several such points, we take $q'$ to be the leftmost. Now, if the rectangle with corners $p$ and $q$ is unsatisfied in $Y'$, then the rectangle with corners $p'$ and $q'$ is unsatisfied in $Y^*$, contradicting the claim that $Y^*$ is a feasible solution.
%%Now consider a non-leaf block $P$ and its skeleton $\tilde P$.
%%It is enough to show that a natural mapping $Y'_P$ of $Y^*_P$ onto $[k] \times [k]$ is feasible for $\tilde P$: Suppose not, then there must be two points $q, q' \in [k] \times [k]$ that are not satisfied. 
%%Suppose $q$ is on the upper-left side of $q'$ (the other case is symmetric).  
%%These two points must correspond to some other points $\phi(q), \phi(q') \in Y^*_P \cup P$; if there are many points, we choose $\phi(q)$ to as close to the bottom-right corner of region $R_q$ as possible. Also we choose $\phi(q')$ to be as close to the upper-left corner of region $R_{q'}$ as possible.
%%Then there must be another point $r$ that satisfies rectangle $\phi(q) \phi(q')$, and such point cannot be inside $R_q \cup R_{q'}$ (otherwise, it would have contradicted the choice of $q$ and $q'$).    
\end{proof} 

\paragraph{Application 3: Recursively decomposable permutations.}
We generalize preorder access sequences to $k$-decomposable permutations, and show that \agreedy has linear access cost for all such permutations.  
Let $k$ be a constant such that there is a tree decomposition $\tset$ in which each node is a permutation of size at most $k$. Notice that a preorder sequence is $2$-decomposable.

\begin{theorem} 
The cost of \agreedy on any $k$-decomposable access sequence is at most $O(n \log k)$. %%%\mayank{1) Remove the at most before the $O$. 2) Shouldn't we state this as a $\Theta$ instead of $O$, as we have the lower bound later? In fact, I think Section H.2 should be moved here.} 
\end{theorem}  
\begin{proof} 
Applying Theorem~\ref{lem: main}, we obtain that $\textsc{RGreedy}(P) \leq 4 \sum_{\tilde{P} \in N(\tset)} \textsc{Greedy}(\tilde{P}) + \sum_{P \in L(\tset)} \textsc{Greedy}(P) +O(n)$. 
First, notice that the elements that belong to the leaf blocks are induced by disjoint intervals, so we have $\sum_{P \in L(\tset)} |P| = n $. 
For each such leaf block, since \greedy is known to have logarithmic amortized cost, we have $\textsc{Greedy}(P) \leq O(|P| \log |P|) = O(|P| \log k)$. Summing over all leaves, we have $\sum_{P \in L(\tset)} \textsc{Greedy}(P) \leq O(n \log k)$.

Next, we bound the cost of non-leaf blocks. 
Since the numbers of children of non-leaf blocks add up to at most $2n$, we have $\sum_{\tilde{P} \in N(\tset)} |\tilde{P}| \leq 2n$. %|L(\tset)| \leq n$. 
Again, using the fact that \greedy has amortized logarithmic cost, it follows that $\sum_{\tilde{P} \in N(\tset)} \textsc{Greedy}(\tilde{P}) \leq \sum_{\tilde{P} \in N(\tset)} O(|\tilde{P}| \log k) \leq O(n \log k)$.
This completes the proof.      
\end{proof} 

We mention that the logarithmic dependence on $k$ is essentially optimal, since there are simple permutations of size $n$ with optimum cost $\Omega(n \log n)$. In fact, for arbitrary $k$, the bound is tight, as shown in Appendix~\ref{sec:tight-block}. In Theorem~\ref{thm:main-noinitial} we proved that for arbitrary constant $k$, the linear upper bound for $k$-decomposable permutations is asymptotically matched by {\sc Greedy}.

\subsection{Proof of Theorem~\ref{lem: main}}

We are now ready to prove the decomposition theorem. Let $Y$ be the output of \agreedy on a permutation of $P$ of size $n$, and let $\tset$ be the block decomposition tree of $P$. We wish to bound $|Y|$. 

For each leaf block $P \in L(\tset)$, let $Y_P$ be the points in  $P \cap Y$.  
For each non-leaf block $\tP \in N(\tset)$, with children $P_1,\ldots, P_k$, let $\tY$ be the set of points in $Y \cap (\tP \setminus \bigcup_{j} P_j)$. In words, these are the points in block ${\tilde{P}}$ that are not in any of its children (we do not want to overcount the cost). It is clear that we have already accounted for all points in $Y$, i.e. $Y = \bigcup_{P \in \tset} Y_P$, so it suffices to bound $\sum_{P \in \tset} |Y_P|$. 

For each non-leaf block $\tP  \in N(\tset)$, we show later that $|\tY| \leq 4 \cdot \Greedy(\tP) + 5\degree(\tP) + \junk(\tP)$, where $\degree(\tP)$ is the number of children of $\tP$, and the 
term $\junk(\tP)$ will be defined such that $\sum_{\tP \in N(\tset)}\junk(\tP) \leq 2n$.
For each leaf block $P$, notice that apart from the \agreedy augmentation in the last row, the cost is exactly $\Greedy(P)$. Including the extra points touched by \agreedy we have that $|Y_P| = \Greedy(P) + |\topwing(P)|$. Since leaf blocks are induced by a partitioning of $[n]$ into disjoint intervals, we have that $\sum_{P \in L(\tset)}{|\topwing(P)|} \leq n$. So in total, the cost can be written as $4 \cdot \sum_{{\tilde{P}} \in N(\tset)}\textsc{Greedy}({\tilde{P}}) + \sum_{P \in L(\tset)}\textsc{Greedy}(P) + 5n$ as desired.
It remains to bound $\sum_{\tilde{P} \in N(\tset)}{|Y_{\tilde{P}}|}$.

Consider any non-leaf block $\tP$. Define the rectangular regions $R(i,j)$ as before. We split $\abs{\tY}$ into three separate parts. Let $\tY^1$ be the subset of $\tY$ consisting of points that are \emph{in the first row} of some region $R(i,j)$. Let $\tY^3$ be the subset of $\tY$ consisting of points that are \emph{in the last row} of some region $R(i,j)$. Observe that any points possibly touched in an \agreedy augmentation step fall into $\tY^3$. Let $\tY^2 = \tY - (\tY^1 \cup \tY^3)$ be the set of touch points that are not in the first or the last row of any region $R(i,j)$. 

Denote by $\textit{left}(R)$ and $\textit{right}(R)$ the leftmost, respectively rightmost column of a region $R$. Similarly, denote by $\textit{top}(R)$ and $\textit{bottom}(R)$ the topmost, respectively bottommost $y$-coordinate of a region $R$.

%%%\laszlo{illustrations badly needed for the proof of the following two lemmas and the claim afterwards.}

\begin{lemma}
\label{lem: onlytwo}
If region $R(i,j)$ has no block $P_q$ below it, \agreedy touches no point in $R(i,j)$. 
If $R(i,j)$ has a block $P_q$ below it, \agreedy touches only points in columns $\textit{left}(R(i,j))$ and $\textit{right}(R(i,j))$.
\end{lemma}

\begin{proof}
If $R(i,j)$ has no block $P_q$ below it, then it is never touched. Assume therefore that $R(i,j)$ is above (and aligned with) block $P_q$ for some $q<j$. When \agreedy accessed the last element in $P_q$, it touched the left and right upper corners of $P_q$ in the augmentation step (since the topmost element in the leftmost and rightmost columns belong to $\topwing(P_q)$). As there are no more accesses after $\textit{top}(P_q)$ inside $[\textit{left}(P_q)+1, \textit{right}(P_q)-1]$, all elements in this interval are hidden in $[\textit{left}(P_q)+1, \textit{right}(P_q)-1]$, and therefore not touched again.
\end{proof}

Let $m_{i,j}$ be the indicator variable of whether the region $R(i,j)$ is touched by \agreedy. 
Consider the execution of \greedy on permutation $\tP$ and define $n_{i,j}$ to be the indicator variable of whether the element $(i,j)$ is touched by \greedy on input $\tilde{P}$. Notice that $\sum_{i,j} n_{i,j} =  \Greedy(\tP)$. We refer to the execution of \greedy on permutation $\tP$ as the ``execution on the contracted instance'', to avoid confusion with the execution of \agreedy on block $\tP$.

\begin{lemma}
\label{lem: greedy-rgreedy} Before the augmentation of the last row of $\tP$ by $\topwing(\tP)$ the following holds:
A rectangular area $R(i,j)$ in block $\tilde{P}$ is touched by \agreedy iff \greedy touches $(i,j)$ in the contracted instance $\tilde{P}$; in other words, $m_{i,j} = n_{i,j}$ for all $i,j$. 
\end{lemma}

Lemma~\ref{lem: onlytwo} and Lemma~\ref{lem: greedy-rgreedy} together immediately imply the bounds $|\tY^1| \le 2 \cdot \Greedy(\tP)$, and $|\tY^3| \le 2 \cdot \Greedy(\tP) + |\topwing(\tP)|$, since in the first or last row of a region $R(i,j)$ that is not one of the $P_j$'s at most two points are touched and no point is touched if \greedy does not touch the region. In the case of $\tY^3$, \agreedy also touches the elements in $\topwing(\tP)$ after the last access to a point in $\tP$. 

\begin{proof}
We prove by induction on $j$ that after $P_j$ is processed, the rectangular area $R(i,j)$ is touched by \agreedy if and only if $(i,j)$ is touched by \greedy in the contracted instance.

We consider the processing of $P_j$ and argue that this invariant is maintained. When $j=1$, the induction hypothesis clearly holds.
There are no other touch points below the block $\tP$ (due to the property of the decomposition), so when points in $P_1$ are accessed, no points in $\tilde{P} \setminus P_1$ are touched. Similarly, in the contracted instance when the first point in $\tilde{P}$ is accessed by \greedy, no other points are touched. 

First we prove that \greedy touching $(i,j)$ in the contracted instance implies that \agreedy touches $R(i,j)$.
Suppose \greedy touches point $(i,j)$. 
Let $p_j$ be the access point in the contracted instance at time $j$. There has to be an access or touch point $(i,j')$ where $j'<j$ such that the rectangle with corners $(i,j')$, and $p'$ had no other points before adding $(i,j)$ (the point $(i,j')$ was on the stair of $p_j$ at time $j$). The induction hypothesis guarantees that $R(i,j')$ is touched by \agreedy, so some point $q$ is touched or accessed in this region.
Choose $q$ to be the topmost such point, and if there are more points in the same row, choose $q$ to be as close to block $P_j$ as possible.

We claim that $q$ is in the stair of the first access point $p^*$ of $P_j$, therefore touching $R(i,j)$: If $q$ was not in the stair, there would be another point $q'$ in the rectangle formed by $p^*$ and $q$. 
Observe that $q'$ cannot be in $R(i,j')$, because it would contradict the choice of $q$ being closest to $p^*$. So $q'$ is in some region $R(i'',j'')$ that lies inside the minimal  rectangle embracing $R(i,j')$, and the region immediately below $P_j$. By induction hypothesis, the fact that $R(i'',j'')$ contains point $q'$ implies that $(i'',j'')$ is touched by \greedy in the contracted instance, contradicting the claim that the rectangle with corners $(i,j')$, $p_j$ is empty.

Now we prove the other direction. Suppose \agreedy touches region $R(i,j)$. This also implies that \greedy (without augmentation step) already touches $R(i,j)$.  But then \greedy must touch a point in $R(i,j)$ when accessing the first element $p^*$ in $P_j$. Assume w.l.o.g.\ that $R(i,j)$ is to the right of $P_j$. It means that there is a point $q$ in the stair of $p^*$, such that $q$ is in some region $R(i,j')$ for $j' < j$.
If in the contracted instance $(i,j')$ is in the stair of the access point $p_j$ at time $j$, we would be done (because \greedy would add point $(i,j)$), so asssume otherwise that this is not the case. Let $(i'',j'')$ 
be the point that blocks $(i,j')$, so the corresponding region $R(i'',j'')$ lies inside the minimal rectangle embracing both $P_j$ and $R(i,j')$.  Also note that the region $R(i'',j'')$ cannot be vertically aligned with $P_j$ because $\tilde P$ is a permutation. Thus $i < i''$.  
The fact that $(i'',j'')$ is touched by \greedy implies, by induction hypothesis, the existence of a point $q' \in R(i'',j'')$, and therefore the augmentation step of \agreedy touches some point $q^*$ on the top row of region $R(i'',j'')$. This point would prevent $q$ from being on the stair of $p^*$, a contradiction.
\end{proof}

\begin{lemma}$\abs{\topwing(\tP)} \le 2 \cdot \degree(\tP)$.\end{lemma}
\begin{proof} $\topwing(\tP)$ contains at most 2 points in block $R(i,k)$ for every $k$. \end{proof}

The only task left is to bound $|\tY^2|$, the number of touch points in non-first and non-last rows of a region. 
This will be the \textit{junk} term mentioned earlier, that sums to $2n$ over all blocks $\tilde{P} \in N(\tset)$. In particular we show that every access point can contribute at most $2$ to the entire sum $\sum_{\tP \in N(\tset)} {|\tY^2|}$.

Consider any access point $X = (x,t_x)$, and let $P^*$ denote the smallest block that contains $X$, and in which $X$ is \emph{not} the first access point (such a block exists for all but the very first access point of the input). Further observe, that every point touched at time $t_x$ \emph{inside} $P^*$ is accounted for already in the sum $\sum_{\tP \in N(\tset)} |Y_{\tilde{P}}^1|$, since these points are in the first row of some region inside $P^*$. It remains to show that at time $t_x$ at most two points are touched outside of $P^*$, excluding a possible \agreedy augmentation step (which we already accounted for). Let $p^*$ denote the first access point in $P^*$, i.e.\ at time $\textit{bottom}(P^*)$. At the time of accessing $p^*$ several points might be touched. Notice that all these points must be outside of $P^*$, since no element is touched before being accessed. Let $p_{\mleft}$ and $p_{\mright}$ be the nearest of these touch points to $p^*$ on its left, respectively on its right.

We claim that for all the access points in $P^*$ other than the first access point $p^*$, only $p_{\mleft}$ and $p_{\mright}$ can be touched outside of $P^*$. Since, by definition $X$ is a non-first element of $P^*$, it follows that at time $t_x$ at most two points are touched outside of $P^*$, exactly what we wish to prove.

Indeed, elements in the interval $[1, p_{\mleft}-1]$, and in the interval $[p_{\mright}+1, n]$ are hidden in $[1, p_{\mleft}-1]$, respectively in $[p_{\mright}+1, n]$, in the time interval $[\textit{bottom}(P^*)+1, \textit{top}(P^*)]$ and therefore cannot be touched when accessing non-first elements in $P^*$. The case remains when $p_{\mleft}$ or $p_{\mright}$ or both are nonexistent. In this case it is easy to see that if on the corresponding side an element outside of $P^*$ were to be touched at the time of a non-first access in $P^*$, then it would have had to be touched at the time of $p^*$, a contradiction.

This concludes the proof.

\section{Proof of \Cref{thm:perm-all}}
\label{sec:perm-enough}

In this section, we prove the following theorem.

\begin{theorem}[\Cref{thm:perm-all}]
	Assume that there is a $c$-competitive online geometric BST $\cA^{p}$ for
	all permutations. %, which is capable of handling insertions. 
	Then there is a $4c$-competitive online geometric
	BST $\cA$ for all access sequences.
\end{theorem}
Given a $c$-competitive online BST $\cA^p$ that can serve a permutation access sequence, we construct 
a $4c$-competitive online BST $\cA$ that can serve any (not necessarily permutation) access sequences.
We first sketch the outline of the proof.

Let $X\in[n]^{m}$ be a (possibly non-permutation) access sequence, and let $T$ be the initial tree on $[n]$. 
First we show that given any online BST $\cB$ that can handle accesses, 
there is a canonical way to augment $\cB$ and obtain an \emph{augmented} $\cB^{\mathit{ins}}$, that can also handle insertions of elements.

Then, we show that given any access-only sequence $X \in [n]^m$, initial tree $T$ and
an online augmented BST ${\cA}^{p,\mathit{ins}}$ that can serve permutation access and insert sequences,
we can construct a permutation sequence (comprising of both accesses and insertions) $X^{p,\mathit{ins}} \in [m]^m$ 
to be fed to $\cA^{p,\mathit{ins}}$ in an online manner. Then there are two steps.
First, we show that the execution of $\cA^{p,\mathit{ins}}$  can define the execution of $\cA$ that runs on $X$ with initial tree $T$.
Second, we show that the cost of $\cA^{p,\mathit{ins}}$ can be bounded by $O(\opt(X))$.
Now we formally describe the proof.

\paragraph{Handling Insertion}
Given an online BST $\cB$, let $\cB^{\mathit{ins}}$ be an algorithm that works as follows. 
To access element $x$, just simulate $\cB$. To insert element $x$, 
try searching for $x$ and miss, i.e.\ touching the search path of predecessor and successor of $x$ in the tree. Then add $x$ as a leaf. Then simulate $\cB$ as if $\cB$ is accessing the leaf $x$.

\subsection{Constructing the algorithm $\cA$}

In the following, we work on $\mathbb{R}\times\mathbb{Z}$ instead
of the integer grid $\mathbb{Z}\times\mathbb{Z}$ in order to simplify
the algorithm and the proof. We extend some notations and definitions
accordingly. We use the words ``$y$-coordinate'' and  ``row'' interchangeably. For any $X,T\subset\mathbb{R}\times\mathbb{Z}$,
by $\left[\begin{smallmatrix}
X\\
T
\end{smallmatrix}\right]$, we mean the point set obtained by ``putting $X$ above $T$''
while the $x$-coordinate of each element does not change. A point
set $X\subset\mathbb{R}\times\mathbb{Z}$ is a \emph{permutation
}if there is at most one point in $X$ for any fixed $x$-coordinate
and any fixed $y$-coordinate. A satisfied set is defined in the
usual no-empty-rectangle way. An online BST $\cB$, given $\left[\begin{smallmatrix}
X\\
T
\end{smallmatrix}\right]\subset\mathbb{R}\times\mathbb{Z}$, outputs a satisfied $\left[\begin{smallmatrix}
\cB_{T}(X)\\
T
\end{smallmatrix}\right]$ in an online manner.

To describe the algorithm $\cA$ for \Cref{thm:perm-all}, we define
two operations: $\Split$ and $\merge$. Given any access sequence
$X$, we define $\Split(X)$ as follows. Let $\epsilon$ be an extremely small positive constant.
For each point $(x,y)\in X$, we have $(x+\epsilon y,y)\in\Split(X)$ except
if $(x,y)$ is the first access point of element $x$ (i.e. there
is no $(x',y)\in X$ where $x'<x$), in which case $(x,y)\in\Split(X)$. The
intention is to ``tilt'' the non-first accesses of all elements
in the simplest way. Note that $\Split(X)$ is a permutation.
We define the reverse operation $\mathit{merge}(S)=\{(\left\lfloor x\right\rfloor ,y)\mid(x,y)\in S\}$
for any set $S$. 
For any (non-permutation) sequence $X\in \mathbb{Z}\times \mathbb{Z}$, let $X^{p,\mathit{ins}},X^p = \Split(X)$ as a point set.  All points in $X^p$ are access points. However, in $X^{p,\mathit{ins}}$, only points with integral $x$-coordinate are access points, the remaining points are insertion points. That is, the first points of each element of $X$ ``become'' access points in $X^{p,\mathit{ins}}$, the remaining points ``become'' insertion points.

Given $\left[\begin{smallmatrix}
X\\
T
\end{smallmatrix}\right] \in \mathbb{Z}\times \mathbb{Z}$ as an input to $\cA$, we define:

$$ \cA{}_{T}(X)=\merge(\cA_{T}^{p,\mathit{ins}}(X^{p,\mathit{ins}}))$$

\begin{lemma}
	\label{lem:A' online}If $\cA^{p}$ is online, then $\cA$ is an online
	process.\end{lemma}
\begin{proof}
	Observe that $\Split$ and $\merge$ fix $y$-coordinate. If $\cA^{p}$
	is online, then $\merge(\cA_{T}^{p,\mathit{ins}}(X^{p,\mathit{ins}}))$ can be computed
	row by row.\end{proof}
\begin{lemma}
	\label{lem:A' sat}For any (non-permutation) access sequence $X$,
	and initial tree $T$ it holds that $\cA{}_{T}(X)\supseteq X$ and $\left[\begin{smallmatrix}
	\cA{}_{T}(X)\\
	T
	\end{smallmatrix}\right]$ is satisfied \end{lemma}
\begin{proof}
	To show that $\cA{}_{T}(X)\supseteq X$, observe that $X=\merge(\Split(X))$
	as $\epsilon$ is very small. Next, $X^{p,\mathit{ins}} \subseteq{\cal A}_{T}^{ p,\mathit{ins} }( X^{p,\mathit{ins}} )$
	because $\cA^{p,\mathit{ins}}$ is an online BST. Finally, as $\merge$ is monotone
	under inclusion, we have
	\[
	X=\merge(\Split(X)) = \merge( X^{p,\mathit{ins}} ) \subseteq\merge(\cA_{T}^{p,\mathit{ins}}( X^{p,\mathit{ins}} ))=\cA{}_{T}(X).
	\]

	To show that $\left[\begin{smallmatrix}
	\cA{}_{T}(X)\\
	T
	\end{smallmatrix}\right]$ is satisfied, we claim that $\merge$ preserves satisfiability. Therefore,
	as $\left[\begin{smallmatrix}
	\cA_{T}^{p,\mathit{ins}}( X^{p,\mathit{ins}} )\\
	T
	\end{smallmatrix}\right]$ is satisfied, $\merge(\left[\begin{smallmatrix}
	\cA_{T}^{p,\mathit{ins}}( X^{p,\mathit{ins}} )\\
	T
	\end{smallmatrix}\right])=\left[\begin{smallmatrix}
	\cA{}_{T}(X)\\
	T
	\end{smallmatrix}\right]$ is satisfied.
	
	It remains to prove the claim. Let $S\subset\mathbb{R}\times\mathbb{Z}$
	be any satisfied set and $S'=\merge(S)$. Consider any points $p',q'\in S'$
	where $p'.x<q'.x$ and choose the rightmost (leftmost) associated
	point $p$ ($q$) of $p'$ ($q'$) from $S$. As $S$ is satisfied,
	there is a point $r\in S\setminus\{p,q\}$ where $r\in\square_{pq}$.
	In particular, $p.x\le r.x\le q.x$. Let $r'=(\left\lfloor r.x\right\rfloor ,r.y)\in S'$
	be the image of $r$ under $\merge$, then we have $r'\in\square_{p'q'}$
	because $\left\lfloor p.x\right\rfloor \le\left\lfloor r.x\right\rfloor \le\left\lfloor q.x\right\rfloor $
	and $\merge$ does not affect $y$-coordinate. Moreover, $r'\neq p',q'$,
	otherwise $p$ is not the rightmost associated point of $p'$, or
	similarly for $q'$. So we have identified $r'\in S'\setminus\{p',q'\}$
	where $r'\in\square_{p'q'}$ and therefore $S'$ is satisfied.
\end{proof}
By \Cref{lem:A' online} and \Cref{lem:A' sat}, ${\cal A}$ is an online
process where $\cA{}_{T}(X)\supseteq X$ and $\cA{}_{T}(X)$ is satisfied.
Therefore, ${\cal A}$ is an online (geometric) BST for accessing arbitrary (non-permutation) sequences.

\subsection{Competitiveness}
The idea is to bound $\cA^{p,\mathit{ins}}_T(X^{p,\mathit{ins}})$, which is the cost of the online BST 
which does both access and insert by the cost of the online BST running on $X$ that only supports access.

Let $T^p$ be a binary search tree with $m$ elements corresponding to the distinct $x$-coordinates of points in $X^p$.
The elements with integral values in $T^p$ form a BST with exactly same structure as $T$. 
Between any two integral value elements $v$ and $v+1$, there will be a right path (each node has only right child) of size 
equal to the number of points from $X^p$ with $x$-coordinate from $(v,v+1)$.
This tree is defined so that the following claim holds.
\begin{lemma}
	\label{lem:bounded by access-only}$\cA^{p,\mathit{ins}}_T(X^{p,\mathit{ins}}) = \cA^{p}_{T^p}(X^p)$.\end{lemma}
\begin{proof}
	We assume that $\cA^p$ does not touch elements in the initial tree 
	that have never been in any of the search paths for access or insert.
	By induction on time step $t$, if the point $(x,t)\in X^{p,\mathit{ins}}$ is an access point, then $\cA^{p,\mathit{ins}}$ and $\cA^{p}$ do the same thing by definition. 
	If the point $(x,t)\in X^{p,\mathit{ins}}$ is an insertion point, then after $\cA^{p,\mathit{ins}}$ adds $x$ as a leaf, the structure of the tree excluding elements which are never touched is same. So  $\cA^{p,\mathit{ins}}$ and $\cA^{p}$ do the same thing again.
\end{proof}

\begin{lemma}
	\label{lem:split is safe}For any access sequence (no insertions) $X$, and the corresponding access sequence $\Split(X)$, $\opt(\Split(X))\le4 \opt(X)$.\end{lemma}
\begin{proof}
	We show that given $\opt(X)$, we can construct a satisfied set
	$S'$ of $\Split(X)$ where $|S'|\le2\opt(X)+2m$. 
	
	We iteratively modify $\opt(X)$ to obtain $S'$. Consider any
	column $c_{i}$ of $\opt(X)$ which contains more than one access
	point, say $m_{i}$ access points. Suppose that $c_{i}$ contains
	$p_{i}\ge m_{i}$ touch points (thus $p_{i} - m_{i}$ non-access touch points), and lies between $c_{i-1}$ and $c_{i+1}$.
	We ``split'' $c_{i}$ into $m_{i}$ columns $c_{i_{1}},c_{i_{2}},\dots c_{i_{m_{i}}}$
	lying between $c_{i-1}$ and $c_{i+1}$. Columns $c_{i_{1}}$ and $c_{i_{m_{i}}}$ have $p_{i}$ points in the same rows as $c_{i}$; in $c_{i_{1}}$ the lowest of them is an access point, and all others above are touched, whereas in $c_{i_{m_{i}}}$ the highest is an access point, and all others below are touched. The remaining columns $c_{i_{2}},\cdots,c_{i_{m_{i}-1}}$ have two points each in them, one access and one touch point. The access point in column $c_{i_{j}}$ is in the same row as the $j$th access in original column $c_{i}$, and the touch point in column $c_{i_{j}}$ is in the same row as the $(j+1)$th access in  $c_{i}$. See \Cref{fig:split}. 
	\begin{figure}
		\centering
		\includegraphics[width=0.6\textwidth]{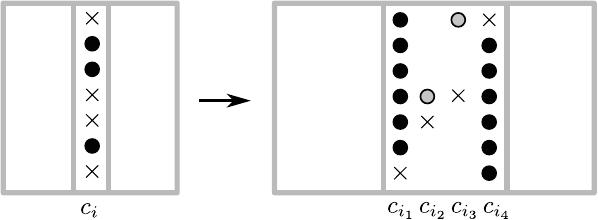} 
		\caption{Illustration of $S'$ in the proof of \Cref{lem:split is safe}. The symbols $\times$ denote accesses and circles are touched points.}
		\label{fig:split}
	\end{figure}

	It is easy to verify that, after placing these points, the set is
	still satisfied. The number of points placed in these columns is exactly
	$2p_{i}+2(m_{i}-2) $. Thus after applying this modification to every
	column $c_{i}$, there are at most $2\opt(X)+2m\le4\opt(X)$ points. 
\end{proof}
Now we are ready to prove the main theorem.

\begin{proofof}{ \Cref{thm:perm-all}}
	First, we have that $ \cA_{T}(X)  \le  \cA_{T}^{p,\mathit{ins}}( X^{p,\mathit{ins}}) $
	because $\merge$ never increases the number of points. 
	We conclude that
	\[
	\cA_{T}(X)  \le  \cA_{T}^{p,\mathit{ins}}( X^{p,\mathit{ins}} ) = \cA_{T^p}^p( X^p )  \leq c\opt(\Split(X)) \le4c\opt(X)
	,\]
	where the second inequality follows from \Cref{lem:bounded by access-only}, 
	the third inequality follows from the assumption that $\cA^p$ is $c$-competitive for only accessing on only permutations, 
	and the last inequality follows from \Cref{lem:split is safe}, and hence $\cA$ is a $4c$-competitive
	online BST for any access sequence.
	\end{proofof}
 
\section{Random permutation is hard}
\label{sec:random-hard}

We prove that the cost of most permutations is $\Omega(n\log n)$ for \emph{all} BST algorithm. A similar result has been shown by Wilber~\cite{wilber} for random access sequences (that might not be permutations).
% Note that the straightforward information-theoretic argument will not work since it will only say that there exists a permutation of cost $\Omega(n \log n)$. 
More formally, we prove the following theorem. 

\begin{theorem} 
\label{thm:random-perm-hard}
There is a positive constant $c$ such that for any $L > 0$ only a fraction $1/2^L$ of the permutations $X$ have
$\opt(X) \le (\log n! - L -n -1)/c$. 
\end{theorem}   
\begin{proof}
We use Kolmogorov complexity. Observe that the shape of a tree on $k$ nodes can be encoded by $O(k)$ bits, e.g., by the preorder traversal. If the keys stored in the tree are clear from the context, the assignment of the keys to the nodes is unique. Let $T$ be the initial tree, let $X$ be any (permutation) access sequence, and let $C_X$ be the total cost of $\opt$ on $X$. We encode the shape of the initial tree with $O(n)$ bits.  Consider now an access of cost $k$. It replaces an initial segment of the current search tree, the before-tree, by another tree on the same set of nodes, the after-tree. The before-tree must include the element accessed and has size $O(k)$. We encode the before-tree and the after-tree with $O(k)$ bits and use a special marker to designate the accessed element in the before-tree. In total, we need $c \cdot (n + C_X)$ bits for the encoding for some small constant $c$. Note that $X$ can be reconstructed from the encoding. We first reconstruct the shape of the initial tree. Its labelling by keys is unique. 
Consider now any access. Since the position of the accessed element in the before-tree is marked and since the before- and the after-tree are on the same set of keys, the search tree after the access is determined. We conclude that we have an injective mapping from permutations $X$ to bitstrings of length $c \cdot (n + C_X)$. 
Assume $C_x \le (\log n! - L - n - 1)/c$. Then the encoding of $X$ has at most $\log n! - L - 1$ bits and hence $X$ is in a set of at most $n!/2^L$ permutations.
\end{proof}

\section{Other claims}

\subsection{Counterexamples for straightforward attempts}

\subsubsection*{Avoiding the same pattern is not enough.} 
\label{sec:bad-example0}

It might be tempting to conjecture that if $X$ avoids $P$, then the \greedy touch matrix $\Greedy(X)$ avoids $P$ as well ($P$ refers both to a permutation pattern and its matrix representation). While this holds for the simple case $P=(2,3,1)$ in the ``no initial tree''-case (proof left as an exercise for the reader), in general it is not true. Figure~\ref{fig:pattern-counter} shows a counterexample ($P = (3,2,1)$).

\begin{figure}[h]
\centering
\includegraphics[scale=1.2]{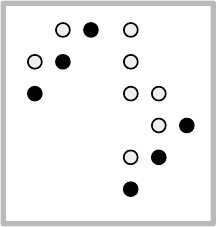}
\caption{Access points are black circles, points touched by \greedy are gray circles. Input avoids (3,2,1), but output contains (3,2,1).}
\label{fig:pattern-counter}
\end{figure}

\subsubsection*{\greedy is not decomposable.}
\label{sec:bad-example1}

One could prove a decomposition theorem similar to Theorem~\ref{lem: main} for \greedy, if \greedy were \emph{decomposable}.
By this we mean that for a $k$-decomposable input sequence, considering a certain level $P=\tilde{P}[P_1,\dots,P_k]$ of the decomposition tree, \greedy touches a region $R(i,j)$, iff the corresponding point $(i,j)$ is touched by \greedy, when executed on the contracted instance $\tilde{P}$. Alas, this property does not hold for \greedy, as shown on Figure~\ref{fig:decomp-counter}.

\begin{figure}[h]
\centering
\includegraphics[scale=1.3]{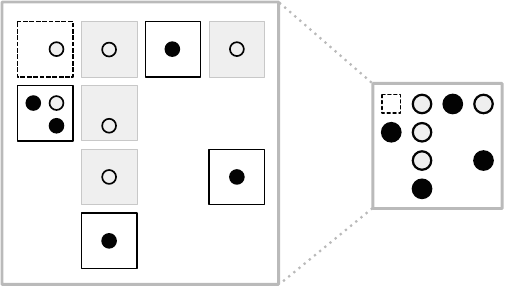}
\caption{Access points are black circles, points touched by \greedy are gray circles. Left: block decomposition of input, and \greedy execution. Top left region is touched. Right: \greedy execution on contracted instance. Top left point not touched.}
\label{fig:decomp-counter}
\end{figure}

\subsubsection*{No permutation input-revealing gadget with initial trees.}
\label{sec:bad-example2}

A permutation input-revealing gadget for \greedy would allow us to strengthen the bound in Theorem~\ref{thm:pattern-avoidance} from quasilinear to linear. In Figure~\ref{fig:gadget-counter} we sketch a construction that shows that such a gadget can not exist, if initial tree is allowed, even if the input is the preorder sequence of a path.

\begin{figure}[h]
\centering
\includegraphics[scale=0.7]{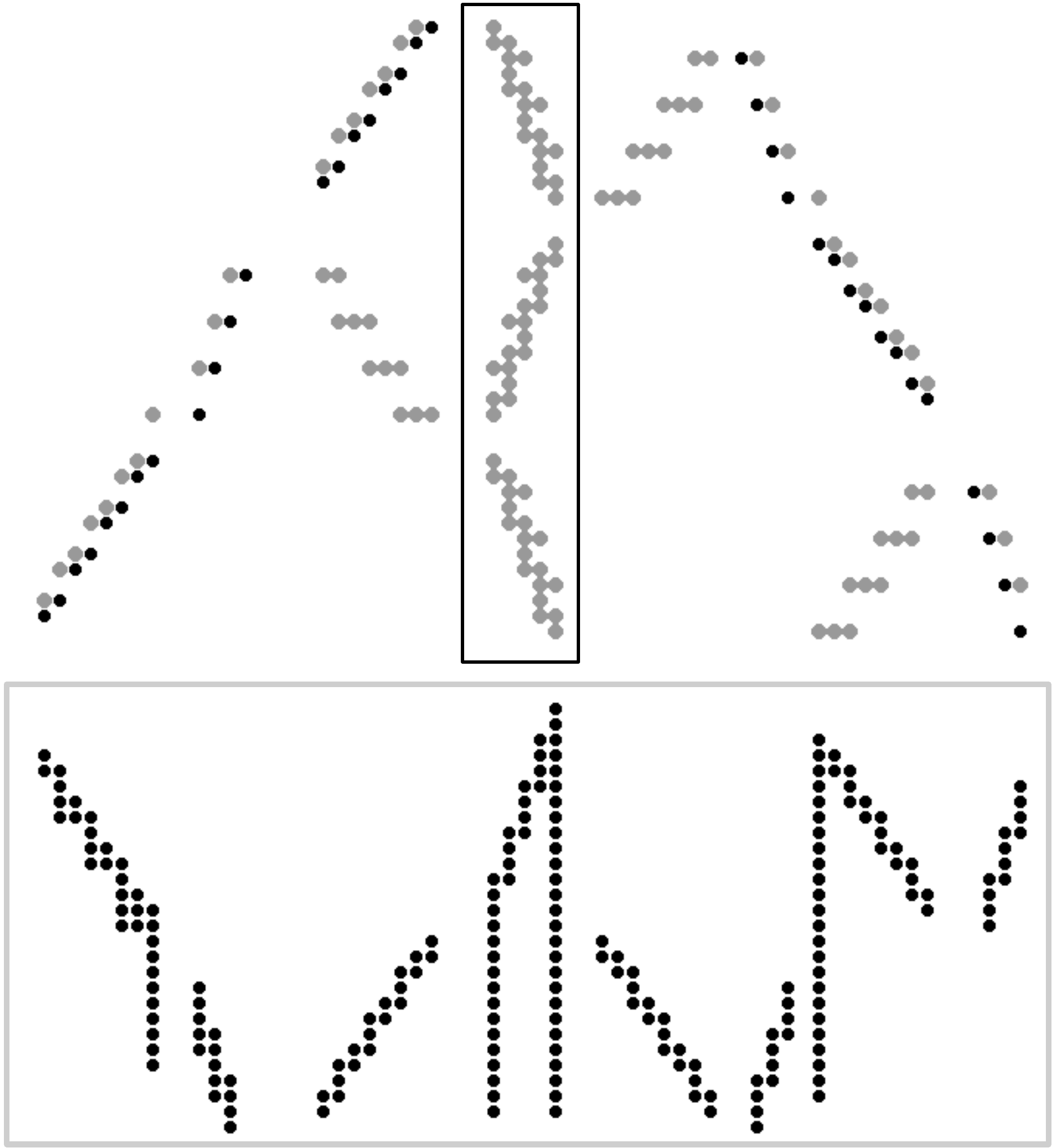}
\caption{A \greedy execution with initial tree. Access points and initial tree are shown with black circles, points touched by \greedy are gray circles. Black rectangle marks a region in which all permutations are contained, but there is no access point inside. Gray rectangle marks initial tree.}
\label{fig:gadget-counter}
\end{figure}

\subsection{Tightness of $O(n \log k)$ for $k$-decomposable permutations}
\label{sec:tight-block} 
In this section, we show that the upper bound of $\opt(X) = O(n \log k)$ for $k$-decomposable permutations is tight, by arguing that there is a sequence that matches this bound. 
We will use the following claim, which is proved in \S\,\ref{lem:decomp opt lower-bound}, and restated here. 

\begin{lemma} 
Let $X$ be any input sequence and $\tset$ be its tree decomposition. Then 
\[\opt(X) \geq \sum_{P \in \tset} \opt(P) -  2n\]
\end{lemma} 

Let $P'$ be a sequence of size $k$ where $\opt(P') = \Omega(k \log k)$. 
Then we construct the tree $\tset$ which is a complete degree-$k$ tree of depth $d$, with each node $v \in \tset$ having the template $P_v = P'$.  
Then $n = k^d$, and each non-leaf node $v \in N(\tset)$ has $\opt(P_v) = \Omega(k \log k)$.  
There are $k^{d-1} = n/k$ non-leaf nodes, so we have $\opt(X) \geq \Omega(k \log k) \cdot \frac{n}{k} = \Omega(n \log k)$.

\subsection{Decomposition of $k$-increasing sequences}
\label{sec:decomposition_monotone}

Given a permutation $X = (X_1, \dots,X_n)$, a subsequence of $X$ of size $k$ is a sequence $(X_{i_1}, \dots, X_{i_k})$, where $1 \leq i_1 < \cdots < i_k \leq n$.

\begin{lemma}
\label{lem:decomposition_monotone}
Let $X \in S_n$ be an arbitrary permutation. The following statements are equivalent:
\begin{compactenum}[(i)] 
\item $X$ is $(k,\dots,1)$-free,
\item There exist pairwise disjoint subsequences $Y_1, Y_2, \dots, Y_{k-1}$ of $X$, such that $Y_i$ is increasing for all $i=1,\dots,k-1$, and $|Y_1| + \cdots + |Y_{k-1}| = n$.
\end{compactenum}
\end{lemma}
\begin{proof}
(ii) $\implies$ (i):\\
Let $Y_1, \dots, Y_{k-1}$ be increasing subsequences of $X$, as in the statement of the lemma, and suppose there exists a subsequence $X' = (X_{i_1}, \dots, X_{i_k})$ of $X$ order-isomorphic to the pattern $(k,\dots,1)$. Since $X'$ is decreasing, no two elements of $X'$ can be in the same subsequence $Y_i$. This is a contradiction, since we have only $k-1$ subsequences. Therefore, such an $X'$ cannot exist, and $X$ is $(k,\dots,1)$-free.

(i) $\implies$ (ii):\\
Assume that $X$ is $(k, \dots, 1)$-free. We construct the decomposition of $X$ into increasing subsequences $Y_1, \dots, Y_{k-1}$ as follows. We look at the plot of permutation $X$, and we refer interchangeably to elements of $X$ and points in the plot of $X$. Let $Y_1$ be the ``wing'', i.e.\ the points in the plot of $X$ that form an empty rectangle with the top left corner (by this definition, we have $X_1 \in Y_1$). Clearly, $Y_1$ forms an increasing subsequence of $X$. We now remove the elements of $Y_1$ from $X$ and similarly find $Y_2$ as the wing of the remaining points. We repeat the process, thereby finding $Y_1, Y_2, \dots$. We claim that we run out of points before reaching $Y_k$, thus constructing a decomposition of $X$ as required by the lemma.

Suppose otherwise that we have reached $Y_k$, and consider an arbitrary point $X_{i_k} \in Y_k$. Since $X_{i_k}$ was not chosen in $Y_{k-1}$, the rectangle with corners $X_{i_k}$ and the top-left corner has some point from $Y_{k-1}$. Pick such a point, and denote it as $X_{i_{k-1}}$. Observe that $X_{i_{k-1}} > X_{i_k}$, and $i_{k-1} < i_k$. Since $X_{i_{k-1}}$ was not chosen in $Y_{k-2}$, we can pick a suitable $X_{i_{k-2}}$ in $Y_{k-2}$. Continuing in this fashion, we construct the subsequence $X_{i_1},...,X_{i_k}$ of $X$ that forms the forbidden pattern $(k,\dots,1)$, a contradiction.
\end{proof}

\subsection{Comparing easy sequences}
\label{sec:comparisons}

\subsubsection*{Easy sequences not captured by any known bounds.}
\label{sec:easy-examples}
We show a sequence $X$ where $UB(X) = \Omega(n \log n)$, $k(X) = \sqrt{n}$, and $\opt(X) = O(n)$. By $UB$ we denote the ``unified bound''~\cite{unified} that subsumes both the working set, dynamic finger and entropy bounds. At this time, however, $UB$ is not known to be matched by any online BST. $UB$ can be seen as the sum of logs of ``$L_1$ distances''. In other words, $X$ is an easy sequence that is not captured by the known classes of easy sequences studied in the literature.  

Consider a point set $\pset$ that is the perturbation of the $\sqrt{n}$-by-$\sqrt{n}$ grid: Let $\ell = \sqrt{n}$, and  
\[\pset = \set{(i \ell + (j-1), j \ell + i-1): i,j \in [\ell] }.\]
Denote the point $(i\ell+(j-1), j \ell +(i-1))$ by $p(i,j)$.  
It is easy to check that there are no two points aligning on $x$ or $y$ coordinates: If $p(i,j)$ agrees with $p(i',j')$ on one coordinate, it implies that $i= i'$ and $j=j'$.   
Therefore, this point set corresponds to a permutation $X \in S_n$. 
Figure~\ref{fig:examples}(i) illustrates the construction. The following lemma shows that $k(X) = \sqrt{n}$. 

\begin{lemma} 
$X$ contains all patterns $\sigma$ of size at most $\ell$. 
\end{lemma} 
\begin{proof} 
Pattern $\sigma$ corresponds to point set $\set{(\sigma_t, t)}_{t=1}^{k}$, this is order-isomorphic to $p(\sigma_1, 1), p(\sigma_2,2),\ldots, p(\sigma_k, k)$.
\end{proof}

\begin{proposition} 
$UB(X) = \Omega(n \log n)$. 
\end{proposition}

Finally, we show that $\opt(X) = O(n)$ by describing an offline algorithm that achieves this bound.
Our algorithm is similar to \greedy, where we augment some extra steps.
First, partition $[n]$ into $I_1,\ldots, I_{\ell}$ such that $I_j = \set{(j-1)\ell +1, j \ell}$. 
Notice that the access sequence starts by touching elements $p(1,1) \in I_1,p(2,1) \in I_2, \ldots, p(\ell,1) \in I_{\ell}$ and then back to $p(2,1)\in I_1$ and so on. 
This occurs for $\ell$ rounds.    

The algorithm proceeds just like \greedy, except that, when accessing $p(i,j)$,  after touching the minimal elements, it also touches the two boundaries of the interval $I_i$, i.e. $(i-1) \ell +1$ and $i \ell$.   
  
\begin{proposition} 
The algorithm produces a feasible solution. 
\end{proposition}

\begin{proposition} 
The cost of this algorithm is $O(n)$. 
\end{proposition} 

A slightly more involved inductive argument shows that \greedy too has linear cost on $X$.

\begin{proposition} 
$\Greedy(X) = O(n)$. 
\end{proposition} 

\subsubsection*{Pattern avoidance and dynamic finger.} 
\label{sec:df-comp}
We show examples that $k(X)$ and the dynamic finger bound (denoted by $DF(X)$) are not comparable. 
Figure~\ref{fig:examples}(ii) illustrates a sequence $X$ where $DF(X) = O(n)$ and $k(X) = \Omega(\sqrt{n}/ \log n)$. 
First, the sequence $X$ starts with the perturbed grid construction of size $f(n)$-by-$f(n)$ where $f(n) = \sqrt{n} /\log n$; then the sequence continues with a sequential access of length $n - f(n)^2$. 
Notice that, in the first phase of the sequence, the distances between consecutive points are $|X_{i+1} - X_i| = O(f(n))$, except for at most $f(n)$ pairs, which have distance $O(f(n)^2)$. 
In the second phase, the distances are one. So $DF(X) = O(f(n)^2) \log f(n) + O(n) = O(n)$. 
On the other hand, the $f(n)$-by-$f(n)$ grid contains all permutations of size $f(n)$, so we have $k(X) = \Omega(\sqrt{n}/ \log n)$.

Figure~\ref{fig:examples}(iii) illustrates a $2$-recursively decomposable sequence (thus having $k(X) = 3$) whose dynamic finger bound is $\Omega(n \log n)$; to see this, notice that $|X_{i+1} - X_i| \geq n/2$ for $i = 1,\ldots, n/2$, so we have $DF(X) \geq \frac{n}{2} \log (n/2)$.   

\iffalse 
\begin{figure}[h]
\centering
\includegraphics[scale=0.3]{avoiding-linear.pdf}
\label{fig:avoiding-linear}
\caption{An example of the ``perturbed grid'' construction when $\ell=4$.}
\end{figure}
\fi

\subsubsection*{Easy non-decomposable sequence.}
Figure~\ref{fig:examples}(iv) shows the non-decomposable input sequence $X = (n/2, n, 1, n-1, 2, \dots)$ on which \greedy achieves linear cost (to see this, observe that \greedy touches at most three elements in each row).

\begin{figure}[h]
\centering
\includegraphics[scale=0.12]{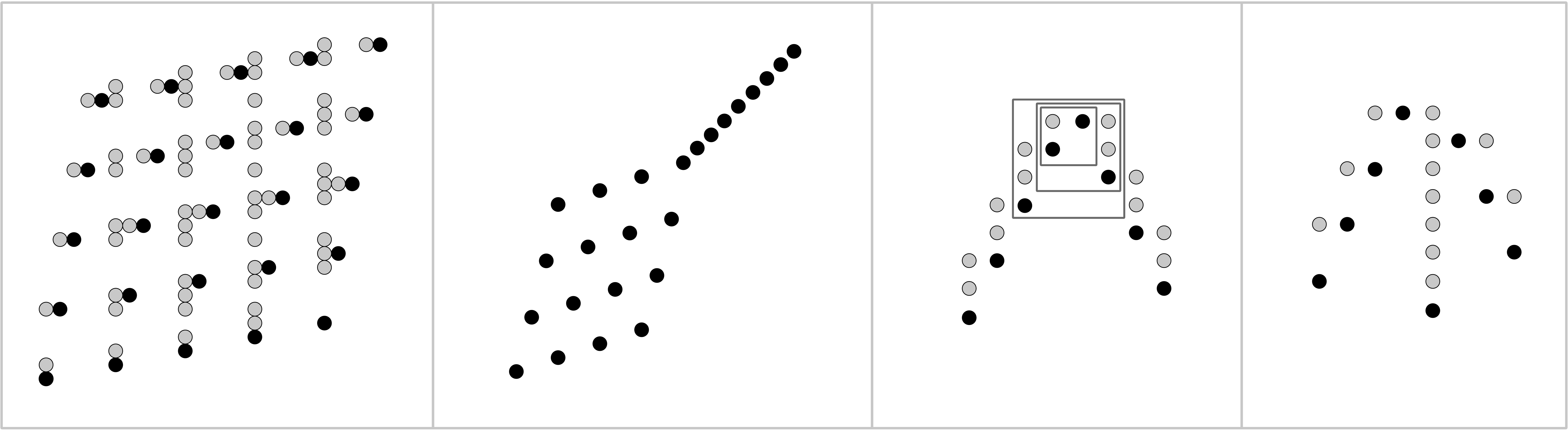}
\caption{From left to right: (i) ``perturbed grid'' construction when $\ell=5$ with \greedy execution, (ii) example with low dynamic finger bound, but high pattern avoidance parameter, (iii) example with low pattern avoidance parameter and high dynamic finger bound, with \greedy execution, (iv) non-decomposable permutation with linear cost \greedy execution.\label{fig:examples}} 
\end{figure}

\subsection{Why is path the roadblock to deque, traversal, split}
\label{sec:path-all}

The deque-, traversal- and split conjectures are originally stated for splay
trees. We first state the conjectures for any online BST $\cA$.
\begin{conjecture}
	[Deque conjecture \cite{tarjan_sequential}]Starting with any initial tree $T$ with $n$ elements,
	the cost of $\cA$ for inserting or deleting the current minimum or maximum elements
	$m$ times is $O(m+n)$
\end{conjecture}

\begin{conjecture}
	[Traversal conjecture \cite{ST85}]Starting with any initial tree $T$ with $n$
	elements, the cost of accessing a sequence $X$ defined by the preorder
	sequence of some BST $T'$ is $O(n)$.
\end{conjecture}

\begin{conjecture}
	[Split conjecture \cite{split_Luc91}]Starting with any initial tree $T$ with $n$ elements,
	the cost of $\cA$ for splitting $T$ by any sequence of $n$ splits is $O(n)$.
	A split at element $x$ is to delete $x$ and to obtain two trees whose
	elements are smaller resp.\ larger than $x$, each subject to
	further splittings.
\end{conjecture}

Next, we state an easier conjecture which is implied by all three
conjectures.
\begin{conjecture}
	[Path conjecture]Starting with any initial tree $T$ with $n$ elements,
	the cost of $\cA$ for accessing a sequence $X$ defined by the preorder sequence
	of some BST $T'$ which is a path is $O(n)$.
\end{conjecture}

To express it in terms of pattern avoidance,
the access sequences in the path conjecture avoid both $(2,3,1)$ and
$(2,1,3)$, while the ones in traversal conjecture only avoid $(2,3,1)$. 

Intuitively, the sequence in path conjecture is a sequence that starts
from minimum or maximum and moves inwards towards the middle, for example,
$(1,2,3,4,10,9,5,6,8,7)$.

\begin{theorem}
	Let $\cA$ be an online BST.
	\begin{enumerate}[label={{(}\roman*{)}}]
\item If $\cA$ satisfies traversal conjecture, then $\cA$ satisfies path
		conjecture. 
		\item If $\cA$ satisfies deque conjecture, then there is an online BST
		$\cA'$ that satisfies path conjecture. 
		\item If $\cA$ satisfies split conjecture, then there is an online BST $\cA'$ that satisfies path conjecture.  		
		
	\end{enumerate}
\end{theorem}
\begin{proof}
		(i) This is clear from the statement.
		
	(ii) If $\cA$ satisfies deque conjecture, then $\cA$ can delete
	the current minimum or maximum element $n$ times with cost $O(n)$.
	Observe that the sequence of these deletions can be defined by a BST
	which is a path. Given $\cA$, we construct $\cA'$ that satisfies path
	conjecture. To access $x=1$ or $n$, $\cA'$ simulates $\cA$ by deleting
	$x$, but then before finishing, $\cA'$ brings $x$ back by making $x$ the root.
	
	To access other elements $x$ from the ``minimum'' (``maximum'')
	side, $\cA'$ also simulates $\cA$, but we can see that it must touch
	$x-1$ ($x+1$) as well. Again, instead of deleting $x$, $\cA'$
	makes $x$ as a root. For any time, the cost of $\cA'$ is at most the
	cost of $\cA$ plus one.

(iii) Splitting is equivalent to deleting when applied to
	the minimum or maximum element. Thus, if $\cA$ satisfies split conjecture, then $\cA$ satisfies deque conjecture without inserts. We can then use the argument of (ii).
	\end{proof}

\end{document}